\documentclass[11pt,a4paper]{article}
\usepackage[margin=1in]{geometry}

\usepackage{amsmath}
\usepackage{autobreak}
\allowdisplaybreaks[4]
\usepackage{amsfonts}
\usepackage{amsthm}
\usepackage{booktabs}
\usepackage{xcolor}
\usepackage{url}
\usepackage{appendix}
\usepackage{overpic}
\usepackage{float}
\usepackage[colorlinks=true, allcolors=blue]{hyperref}
\usepackage{listings}
\usepackage{indentfirst}
\usepackage{diagbox}
\usepackage[ruled, linesnumbered]{algorithm2e}
\usepackage[numbers,sort]{natbib}
\usepackage{url}
\usepackage{hvlogos}

\usepackage{graphicx}
\usepackage{caption}
\usepackage{subcaption}
\usepackage{tikz}
\tikzset{
    node/.style={circle,draw=black},
    edge/.style={draw=black}
}

\DeclareMathOperator*{\argmax}{arg\,max}

\newcommand{\opt}{\textsf{OPT}}
\newcommand{\alg}{\textsf{ALG}}

\renewcommand{\Pr}{\operatorname*{\mathbf{Pr}}}

\usepackage[capitalise]{cleveref}
\newtheorem{mytheorems}{Theorem}
\newtheorem{thm}[mytheorems]{Theorem}
\newtheorem{lem}[mytheorems]{Lemma}
\newtheorem{prop}[mytheorems]{Proposition}

\newtheorem{inv}{Invariant}

\title{Edge Arrival Online Matching: The Power of Free Disposal on Acyclic Graphs}

\author{Tianle Jiang \thanks{Shanghai Jiao Tong University. {\texttt{\{jiangtianle, zhang\_yuhao\}@sjtu.edu.cn}}} \and
Yuhao Zhang\footnotemark[1]}

\date{}

\begin{document}

\maketitle

\begin{abstract}
Online matching is a fundamental problem in the study of online algorithms. We study the problem under a very general arrival model: the edge arrival model. Free disposal is an important notion in the online matching literature, which allows the algorithm to dispose of the previously matched edges. Without free disposal, we cannot achieve any bounded ratio, even with randomized algorithms, when edges are weighted. 

Our paper focuses on clarifying the power of free disposal in both the unweighted and the weighted setting. As far as we know, it's still uncertain if free disposal can give us extra leverage to enhance the competitive ratio in the unweighted scenario, even in specific instances such as Growing Trees, where every new edge adds a new leaf to the graph. Our study serves as a valuable initial exploration of this open question. The results are listed as follows:

1. With free disposal, we improve the competitive ratio for unweighted online matching on Growing Trees from $5/9$ to $2/3 \approx 0.66$, and show that the ratio is tight. For Forests, a more general setting where the underlying graph is a forest and edges may arrive in arbitrary order, we improve the competitive ratio from $5/9$ to $5/8 = 0.625$. 

2. Both the ratios of $2/3$ and $0.625$ show a separation to the upper bound of the competitive ratio without free disposal on Growing Trees ($0.5914$). Therefore, we demonstrate the additional power of free disposal for the unweighted setting for the first time, at least in the special setting of Growing Trees and Forests. 

3. We improve the competitive ratio for weighted online matching on Growing Trees from $1/3$ to $1/2$ using a very simple ordinal algorithm, and show that it is optimal among ordinal algorithms.

\end{abstract}

\section{Introduction}

\noindent Online matching stands as a central problem in the field of online algorithms because of its clean definition and broad applications. These problems come in various forms with respect to the arrival models, edge weights, and objectives. In the seminal work of Karp, Vazirani, and Vazirani \cite{KVV}, they propose the Online Bipartite Matching Problem, where one side of the graph is known upfront, and on the other side, vertices arrive one by one with the information of their adjacent edges. The designed online algorithm needs to decide the matching of the arriving vertex immediately and irrevocably, with the objective of maximizing the total cardinality of the matching. We term this arrival model "one-sided vertex arrival" and the objective "unweighted". Regarding the arrival mode, we can further study it under a more generalized setting, such as the "general vertex arrival" \cite{Gamlath19, WangWong}, where all vertices in the graph become online, and the more general "edge arrival", where we receive edges one by one. On the other hand, considering the objective and the edge weight, we will have the special vertex-weighted variant \cite{vertexweight, vertexweight2}, the edge-weighted variant, and the AdWords variant \cite{Adwords2, Adwords3, Adwords}. 


This paper focuses on the online matching problem under the edge arrival model and discusses both the unweighted and the weighted settings. For simplicity, we call the unweighted problem MCM (Max Cardinality Matching) and the weighted version MWM (Max Weight Matching). An optimal deterministic online algorithm for edge arrival MCM is Greedy, which is $0.5$-competitive. Moreover, \citet{Gamlath19} showed that $0.5$ is the best competitive ratio we can achieve, even for randomized and fractional algorithms. 

If we associate edges with weight, the problem becomes harder. Without making further assumptions, we cannot achieve any bounded competitive ratio in edge arrival MWM, even when the graph is a star (it means that the hardness holds for almost all arrival models). To this end, the notion of "free disposal" (a.k.a., preemption) is introduced. In this setting, the algorithm is allowed to discard any previously matched edge, while a discarded edge cannot be re-matched later. \citet{Feigenbaum05} first studied the edge arrival MWM problem with free disposal and gave a $1/6$-competitive deterministic algorithm. Later, the ratio was improved to $3-2\sqrt 2 \approx 0.1716$ by \cite{McGregor05}, and was proved to be optimal among all deterministic algorithms by \citet{Varadaraja11}.  \citet{Epstein5.356} presented a $0.1867$-competitive randomized algorithm for edge arrival MWM. 

Move back to the unweighted setting, can "free disposal" also be useful for MCM? The hard instance for the edge arrival setting by \citet{Gamlath19} does not hold when the algorithm can use the power of free disposal. The current best upper bound with disposal is $0.585$ by \citet{fully19}, and the current best competitive ratio is only $0.5$ by Greedy. The power of free disposal is still unknown, and whether there is an algorithm with a competitive ratio better than $0.5$ in the edge arrival setting with free disposal remains open. 

In conclusion, there are two main open questions in this line to clarify the power of free disposal in edge arrival online matching. 
\begin{enumerate}
    \item Can we beat $0.5$ with free disposal in unweighted edge arrival online matching?
    \item What is the best competitive ratio for weighted edge arrival online matching?
\end{enumerate}

As a crucial starting point for answering the two open questions above, we focus on a special arrival model called Growing Trees introduced by \citet{Chiplunkar15/28}. In this setting, the input graph is a tree at every time point. In other words, each arriving edge "grows" a leaf to the tree. The algorithmic challenge is quite clear: "Should we benefit the new vertex by hurting an old one?" Although more challenges will be introduced when the arrival model becomes more general, it is worth clarifying the best competitive ratio we can achieve under this special setting, because it has been used to obtain some hardness results. \citet{Buchbinder5/9} constructed a growing tree as hard instance and showed a $\frac{2}{3+1/\phi^2}\approx 0.5914$ upper bound for edge arrival MCM, which also applies to vertex arrival MCM. Moreover, in edge arrival MWM, the $3 - 2\sqrt{2}$ tight hard instance for deterministic algorithms \cite{Varadaraja11} is also a growing tree. Therefore, we believe that algorithmic improvements in this specialized model may offer valuable insights for the more general models.

Unfortunately, even for Growing Trees, the power of free disposal remains unclear in the literature. In the unweighted setting, \citet{Chiplunkar15/28} presented a $15/28\approx 0.535$-competitive randomized algorithm for Growing Trees. \citet{Tirodkar1/3} generalized the Growing Trees model to Forests (allowing an arriving edge to merge two trees), analyzed the same algorithm as \cite{Chiplunkar15/28} and obtained a competitive ratio of $33/64 \approx 0.515$. Subsequently, \citet{Buchbinder5/9} presented a $5/9\approx 0.555$-competitive algorithm for both settings that do not require the ability of free disposal, and no known algorithm with free disposal beats the $0.555$ competitive ratio without free disposal. Therefore, it remains open whether there is a separation between the competitive ratio in the setting of "with free disposal" and the upper bound in the setting of "without free disposal". 

In the weighted setting, the best randomized algorithm for Growing Trees MWM is $1/3$-competitive \cite{Tirodkar1/3}, and there is still a huge gap between the lower and upper bound. In particular, we do not have a stronger upper bound than the unweighted setting.

\subsection{Our Results}

\noindent In this paper, we significantly improve the competitive ratio of MCM with free disposal on Growing Trees to $2/3 \approx 0.66$ and show that it is optimal among all randomized algorithms (even fractional).
Remark that it beats the upper bound of $0.5914$ in the setting without free disposal. Therefore, it shows a separation between the setting "with free disposal" and "without disposal" on Growing Trees. We also generalize our algorithm to the Forests setting and achieve a competitive ratio of $5/8 = 0.625 > 0.5914$, which also shows the strict power of "free disposal" on Forests. 

\begin{thm}
\label{thm:grwoingtreealg}
    There is a $2/3$-competitive randomized algorithm for MCM on Growing Trees.
\end{thm}

\begin{thm}
\label{thm:growingtreehard}
    No randomized (or even fractional) algorithm can achieve a competitive ratio better than $2/3$ for MCM on Growing Trees.
\end{thm}

\begin{thm}
\label{thm:forestealg}
    There is a $5/8$-competitive randomized algorithm for MCM on Forests.
\end{thm}

As an extension of our techniques, we design a very simple $1/2$-competitive algorithm for the edge arrival MWM problem on Growing Trees. Interestingly, our algorithm is "ordinal": it only makes comparisons between edge weights and does not need to known the exact values. Previously, research on online ordinal algorithms focused on the random order arrival model, such as various secretary problems \cite{matroidsecretary,matroidsecretary2,JKsecretary} and the two-sided game of Googol \cite{googol,secretarysample}. We refer readers to \cite{ordinal} for more comprehensive results on this model. Therefore, it is rather surprising that it also applies to our setting with adversarial arrival order. We also prove that $1/2$ is the optimal competitive ratio with respect to ordinal algorithms. 

\begin{thm}
\label{thm:weightedalg}
    There is a $1/2$-competitive randomized ordinal algorithm for MWM on Growing Trees.
\end{thm}

\begin{thm}
\label{thm:weightedhard}
    No randomized (or even fractional) ordinal algorithm can achieve a competitive ratio better than $1/2$ for MWM on Growing Trees. 
\end{thm}

Our results and comparison with previous works are listed in \Cref{tab:results}. Note that all the competitive ratios here hold for randomized algorithms, and all hardness results hold for fractional algorithms and thus for randomized algorithms.

\begin{table}[h]
\caption{Our results (marked in red).}
\label{tab:results}
\begin{tabular}{|l|l|l|l|l|l|}
\hline
               &                       & unweighted            &          & weighted              &           \\ \hline
Settings               & Free Disposal                      & algorithm     & hardness & algorithm     & hardness  \\ \hline
General Graph & without FD & $0.5$                 & $0.5$ \cite{Gamlath19}   & unbounded             & unbounded \\ \hline
               & with FD    & $0.5$                 & $0.585$ \cite{fully19} & $0.1867$ \cite{Epstein5.356}                & $0.585$ \cite{fully19} \\ \hline
Forests         & without FD & $5/9$ \cite{Buchbinder5/9}                & $0.5914$ \cite{Buchbinder5/9}  & unbounded             & unbounded   \\ \hline
               & with FD    & {\color{red} $5/9 \rightarrow 5/8$} & {\color{red} $2/3$}    &  0.1867 \cite{Epstein5.356}                     & {\color{red} $2/3$}     \\ \hline
Growing Trees  & without FD & $5/9$ \cite{Buchbinder5/9}                & $0.5914$ \cite{Buchbinder5/9} & unbounded             & unbounded  \\ \hline
               & with FD    & {\color{red}$5/9 \rightarrow 2/3$} & {\color{red} $2/3$}    & {\color{red}$1/3$ \cite{Tirodkar1/3} $\rightarrow 1/2$} & {\color{red} $2/3$}     \\ \hline
\end{tabular}
\end{table}

\subsection{Overview of Techniques}

\noindent The main technical highlights in our paper are a novel lossless rounding framework for solving this problem and an uncertain dual assignment technique for online primal-dual analysis. 

All previous works in the Growing Trees and the Forests model (\cite{Buchbinder5/9, Chiplunkar15/28, Tirodkar1/3}) maintain a constant number of different matchings deterministically and consistently output each matching with a certain probability. However, this method has a drawback: since the multiple matchings are strongly correlated with each other, the complexity of the analysis grows as the number of matchings increases. For example, the algorithm for MCM on Growing Trees in \cite{Chiplunkar15/28} maintains $4$ matchings, and although it is almost certain that maintaining more than $4$ matchings can achieve better performance, it would be very hard to analyze.

\paragraph{Lossless Rounding Framework.} We take a different approach to obtain improved results. We propose a novel framework for solving this problem. We first design a lossless rounding scheme, which shows that any fractional algorithm for Growing Trees or Forests that satisfies certain additional constraints can be rounded to a randomized integral algorithm with the same competitive ratio. A similar idea has been proposed by \citet{lossless_rounding} for online bipartite matching. Due to the simplicity of the tree structure, there is no correlation between vertices, so we only need a mild condition to promise a lossless rounding. For growing trees, no additional constraint is needed. The algorithm for Forests only needs to satisfy one additional constraint: if an arriving edge $(u,v)$ merges two trees, let $x(u,v)$ denote the fraction we match to $(u,v)$ and $x(u),x(v)$ denote the already matched fraction of $u$ and $v$ respectively, then the algorithm should satisfy $x(u,v) \leq {(1-x(u))(1-x(v))}$. This constraint allows us to successfully match $(u,v)$ with probability $x(u,v)$, since if $u$ and $v$ belong to different trees when $(u,v)$ arrives, the event that $u$ is already matched and $v$ is already matched are independent. 

With the help of lossless rounding, we can focus on designing fractional algorithms, which makes the algorithmic approach conceptually cleaner. Especially for MCM on Growing Trees, we only use a 1-page proof to achieve the optimal competitive ratio of $2/3$. 

\paragraph{Uncertain Dual Assignment.} Our design and analysis of the fractional algorithms mainly build on the online primal dual framework (refer to \cite{pdframework} for details). The main technical task of this framework is a carefully designed dual assignment scheme that maintains a set of dual variables online so that each edge is approximately "covered" (to keep the approximate dual feasibility). 

In many previous works, the dual assignments are usually fixed after we make a matching decision \cite{Devanur13, fully19, fully20, Adwords2}. For example, in the randomized primal dual analysis of the Ranking algorithm \cite{Devanur13}, when we match an offline vertex, the dual assignment is fixed immediately w.r.t. the rank of the offline vertex. However, in our work, we introduce a technique called uncertain dual assignment (for the Forests setting). When we match an edge, we can store some uncertain dual value on the edge, instead of assigning it to a fixed endpoint. This guarantees approximate dual feasibility at the current time. After some future event happens, we gain more knowledge of the graph structure and will then either fix the uncertain dual assignment to one endpoint or dispose of it to comply with the decrement in the primal solution. We note that this technique is not restricted to the tree structure, and may also be applied to other online problems.

\subsection{Related Works}

\noindent The MCM problem has been well studied under different arrival models besides edge arrival. As an extension of the one-sided online bipartite matching, \citet{WangWong} introduced the general vertex arrival model, where vertices on both sides are revealed one at a time, together with edges to previously revealed neighbors. They presented a $0.526$-competitive fractional algorithm. \citet{Gamlath19} presented a $1/2 + \Omega(1)$-competitive randomized algorithm, which essentially rounds online the fractional algorithm by \cite{WangWong}. The current best upper bound for this problem is $0.584$ \cite{fully24}.
Other arrival models in literature include the Fully Online Matching \cite{fully18, fully19, fully20, fully24} and the Online Windowed Matching \cite{windowed}, where each vertex has not only an arrival time but also a departure time.

\citet{Bullinger1/6} studied the online coalition formation problem in the general vertex arrival model, where matching is a special case (each coalition has size $2$). They gave a $1/6$-competitive algorithm, which can actually be improved to $3-2\sqrt 2$ in the same way as \cite{McGregor05}. 
There is a substantial gap between general vertex arrival and one-sided arrival. For MWM with one sided arrival, greedy is $0.5$-competitive, and \citet{weighted0.5086} designed a $0.5086$-competitive randomized algorithm, beating the $0.5$-barrier for the first time. Their algorithm was later improved to $0.519$-competitive by \citet{weighted0.519}. Besides free disposal, there are various other attempts to achieve a constant competitive ratio for online edge-weighted matching. For example, \citet{stochastic0.645} presented a $0.645$-competitive algorithm for edge-weighted online stochastic matching, which does not need free disposal. This result is further improved to $0.650$ by \citet{stochastic0.650}. 

There is relatively little research on lossless rounding for online matching problems. \citet{Devanur13} showed that lossless online rounding for arbitrary fractional bipartite matching is impossible. However, motivated by the fact that each randomized algorithm induces a fractional algorithm with the same competitive ratio, \citet{lossless_rounding} studies the necessary conditions for lossless rounding. They proposed a set of constraints such that any fractional algorithm that satisfies these constraints admits a lossless rounding scheme. Our rounding scheme differs from theirs in that we have to deal with free disposal.

\section{Preliminaries}
\subsection{Problem Definitions}

\paragraph{Edge Arrival Online Matching with Free Disposal.}

Consider an underlying graph $G = (V,E)$. Edges arrive one at a time, and once an edge arrives, the algorithm must immediately decide whether to match or discard that edge. We allow the algorithm to perform "free disposal": after an edge is matched, the algorithm can discard it at any time, possibly to make room for a newly revealed edge. However, if an edge is discarded in this way, it cannot be matched again afterward. In the MCM problem, the goal is to maximize the cardinality of the matching. In the MWM problem, each edge $e=(u,v)$ is associated with a weight that is also revealed upon its arrival, and is denoted by $w(e)$ or $w(u,v)$. The goal is to maximize the total weight of matched edges.

\paragraph{Growing Trees.}

Growing Trees is a special case of the edge arrival model. The underlying graph is a tree, and at any time, the edges that have arrived form a tree. In other words, each arriving edge "grows" a leaf to the tree. For convenience of analysis, we will assume that the tree is rooted at one of the endpoints of the first edge. Then, when a new edge $(u,v)$ arrives, where only $u$ is incident to some previous edge, we call $u$ the parent of $v$, denoted as $p_v$. To be clear, in the context of Growing Trees, when we mention an edge $(u,v)$, we always assume that $u$ (the left vertex) is the parent vertex of $v$ (the right vertex). 

\paragraph{Forests.}

Forests is a slightly more general arrival model than Growing Trees. Now the underlying graph is a forest, and the edges may arrive in any order. In particular, an arriving edge can connect two trees and merge them into one larger tree.

\subsection{Online Primal Dual Framework}

\noindent The Online Primal Dual Framework has been widely used to analyze the competitive ratio of online matching algorithms (\cite{Devanur13, fully19, weighted0.5086, Adwords3}). Consider the LP of (edge-weighted) matching and its dual:
\begin{align*}
\max: \quad & \textstyle \sum_{(u,v)\in E} x(u,v)\cdot w(u,v)&& \qquad\qquad & \min: \quad & \textstyle\sum_{u \in V} \alpha(u)\\
\text{s.t.} \quad & \textstyle \sum_{v:(u,v)\in E} x(u,v) \leq 1 && \forall u\in V & \text{s.t.} \quad & \alpha(u) + \alpha(v) \geq w(u,v) && \forall (u,v)\in E \\
& x(u,v) \geq 0 && \forall (u,v)\in E & & \alpha(u) \geq 0 && \forall u \in V
\end{align*}

Assume we have an algorithm for online matching that always maintains a feasible primal solution $x$. 
\begin{lem}
    The algorithm is $c$-competitive if there is a set of non-negative dual variables $\alpha$ that satisfy the following two conditions at all times:
\begin{itemize}
    \item Reverse weak duality: $P=\sum_{(u,v)\in E}x(u,v)\cdot w(u,v) \geq \sum_{u\in V} \alpha(u)=D$.
    \item Approximate dual feasibility: $\alpha(u)+\alpha(v) \geq c \cdot w(u,v), \forall (u,v)\in E$.
\end{itemize}
\end{lem}

Following the framework, we can transform the task of proving the competitive ratio into constructing an approximate feasible dual solution. 
\section{Fractional Algorithms and A Lossless Rounding}
\label{sec:framework}

\subsection{Fractional Algorithm for Matching with Disposal}

\noindent We use deterministic fractional algorithms as an intermediate step for designing randomized integral algorithms. For each edge $e = (u,v)$, let $x(e)$ or $x(u,v)$ denote the matched fraction of the edge. For each vertex $u$, let $x(u) = \sum_{v: (u,v)\in E}x(u,v)$ denote the total matched fraction of its incident edges. 

For each arriving edge $e = (u,v)$, the fractional algorithm should match a fraction to $e$ under the matching constraint. That is, the algorithm should keep $x(u) \leq 1$ and $x(v) \leq 1$. With the ability to have free disposal, the algorithm can decrease the matched fraction of edges at any time. However, the matched fraction of an edge can only increase at its arrival time. 

\subsection{Lossless Randomized Rounding}

\noindent Next, we present a rounding framework for Growing Trees and Forests. Under some mild conditions, we prove that the rounding is lossless: given a $c$-competitive fractional algorithm $\textsf{A}_F$, we can obtain a $c$-competitive randomized integral algorithm $\textsf{A}_R$. 

When each edge arrives, $\textsf{A}_F$ essentially performs the following two steps sequentially:
\begin{enumerate}
    \item Dispose Step: lower the matched fraction of some previously arrived edges.
    \item Matching Step: match some fraction to the arriving edge.
\end{enumerate}
We treat the two steps separately. For a clear representation, we assume that the $t$-th edge arrives at time $2t-1$. The disposal step is performed at a time of $2t-1$, and the matching step is performed at time $2t$. We introduce the following notations for convenience: 
\begin{itemize}
    \item $M(e,t)$ denotes the event that edge $e$ is matched by $\textsf{A}_R$ right after time $t$.
    \item $M(u,t)$ denotes the event that vertex $u$ is matched by $\textsf{A}_R$ right after time $t$.
    \item $x_t(e)$ denotes the matched fraction of $e$ in $\textsf{A}_F$ right after time $t$.
    \item $x_t(u)$ denotes the total matched fraction of edges incident to $u$ in $\textsf{A}_F$ right after time $t$.
\end{itemize}

To ensure that the rounding is lossless, we maintain the following invariant in our rounding scheme: at any time $t$, for each edge $e$, $\Pr[M(e,t)] = x_t(e)$.
If the invariant holds, the expected size of the matching maintained by $\textsf{A}_R$ always equals the size of fractional matching maintained by $\textsf{A}_F$. Thus the competitive ratio is the same. 

We will prove that this property is always held by induction. If the matched fraction of an edge does not change at time $t$, then the matching state of that edge does not change, and the property still holds trivially. Next, we consider three cases where we change the matched fraction of an edge.

\paragraph{Disposal Step.} For an edge $e$, if $x_{t-1}(e) = \lambda$ and $\textsf{A}_F$ disposes a $\delta$-fraction from $e$ when an edge arrives at time $t$, then $\textsf{A}_R$ performs the following step:
\begin{itemize}
    \item If $e$ is matched, dispose it with probability $p=\frac{\delta}{\lambda}$. Otherwise, do nothing.
\end{itemize}
Therefore, if $\Pr[M(e,t-1)] = x_{t-1}(e) = \lambda$ by induction hypothesis, then
$$
\begin{aligned}
\Pr[M(e,t)] = &~\Pr[M(e,t-1)] \cdot \Pr[\text{$e$ is not disposed}] \\
= &~ \lambda \cdot \left(1 - \frac{\delta}{\lambda}\right)\\
= &~ \lambda - \delta = x_{t}(e) ~.
\end{aligned}
$$

\paragraph{Match an edge that extends a leaf.} In this case, an edge $e = (u,v)$ extends a leaf $v$ to a tree, and $\textsf{A}_F$ matches a $\gamma$-fraction to it at time $t$. Whether to match $e$ depends on whether $u$ is matched. Suppose $x_{t-1}(u) = \mu$, then $\textsf{A}_R$ performs the following step: 
\begin{itemize}
    \item If $u$ is not matched, match $e$ with probability $\frac{\gamma}{1-\mu}$.
\end{itemize}
Any valid fractional algorithm guarantees $\mu + \gamma \leq 1$, so this probability does not exceed $1$. Since at most one edge incident to $u$ can be matched, the probability that $u$ is matched is
$$
\Pr[M(u,t-1)] = \sum_{e'\ni u}\Pr[M(e',t-1)] = \sum_{e'\ni u}x_{t-1}(e') = x_{t-1}(u) ~.
$$
Therefore the probability that $e$ becomes matched at time $t$ is
$$
\begin{aligned}
\Pr[M(e,t)] =&~ \Pr[\overline{M(u,t-1)}] \cdot \Pr[\text{match $e$ at time $t$} \mid \overline{M(u,t-1)}]\\
=&~ (1 - \mu) \cdot \frac{\gamma}{1-\mu}\\
=&~ \gamma = x_t(e) ~.
\end{aligned}
$$

\paragraph{Merge a tree with only one edge with another tree.} Suppose there is a tree with only one edge $e'$, and the arriving edge $e$ merges $e'$ with another tree. In this case, we can assume that $e$ arrives first and then $e'$ arrives, both growing a leaf to the other tree, then calculate the matched fractions. To ensure that the desired fraction can always be achieved, we impose a trivial condition on the fractional algorithm: if an arriving edge is not incident to any previous edge (i.e. a new tree with one edge), the algorithm must match fraction $1$ to the edge. The corresponding step in $\textsf{A}_R$ goes as follows: first decide whether $e$ and $e'$ are matched (assuming $e'$ arrives last), then dispose $e'$ and match $e$ if necessary. 

\paragraph{Match an edge that merges two trees.} In this case, the arriving edge $e = (u,v)$ merges two "non-trivial" trees (with at least $2$ edges), and $\textsf{A}_F$ matches a $\gamma$-fraction to it at time $t$. It only happens in the forest setting. Whether to match $e$ depends on whether $u$ and $v$ are matched.
Following our rounding scheme, these two events are independent. Suppose $x_{t-1}(u) = \mu_1$ and $x_{t-1}(v)= \mu_2$, then $\textsf{A}_R$ performs the following step: 
\begin{itemize}
    \item If $u$ and $v$ are both not matched, match $e$ with probability $\frac{\gamma}{(1-\mu_1)(1-\mu_2)}$.
\end{itemize}
However, a valid fractional algorithm does not necessarily guarantee that $\frac{\gamma}{(1-\mu_1)(1-\mu_2)} \leq 1$. Therefore, it should be an additional constraint in $\textsf{A}_F$. Combining with the independent condition, the probability that $e$ becomes matched at time $t$ is
$$
\begin{aligned}
    \Pr[M(e,t)] = &~ \Pr[\overline{M(u,t-1)} \land \overline{M(u,t-1)}] \cdot \Pr[\text{match $e$ at time $t$} \mid \overline{M(u,t-1)} \land \overline{M(u,t-1)}]\\
    = &~ \Pr[\overline{M(u,t-1)} ] \cdot \Pr[\overline{M(u,t-1)}] \cdot \frac{\gamma}{(1-\mu_1)(1-\mu_2)} \\
    = &~ (1-\mu_1)\cdot (1-\mu_2)\cdot \frac{\gamma}{(1-\mu_1)(1-\mu_2)} \\
    = &~ \gamma = x_t(e) ~.
\end{aligned}
$$

Finally, we formalize our rounding scheme as the following two lemmas:

\begin{lem}
\label{lem:rounding1}
    For online matching with disposal on Growing Trees, if there is a $c$-competitive fractional algorithm, then there is a corresponding $c$-competitive randomized algorithm.
\end{lem}
\begin{proof}
    When we consider growing trees, there is only a disposal step and a matching leaf step. We can keep the lossless invariant without any conditions. 
\end{proof}

\begin{lem}
\label{lem:rounding2}
    For online matching with free disposal on Forests, if there is a $c$-competitive fractional algorithm with the following two conditions:
    \begin{enumerate}
        \item for any isolated edge $e$, $x(e)=1$,
        \item at the matching step $t$ of a new edge $(u,v)$ that merges two trees, $\frac{x_t(u,v)}{(1-x_{t-1}(u))(1-x_{t-1}(v))} \leq 1$,
    \end{enumerate}
    then there is a corresponding $c$-competitive randomized algorithm.
\end{lem}
\begin{proof}
    Compared to growing trees, there may exist edges that connect two different trees. However, when the two conditions hold, we have also shown that $\Pr[M(e,t)] = x_t(e)$ can be satisfied. 
\end{proof}

\section{Online MCM on Growing Trees}
\label{sec:unweighted}

\noindent In this section, we first present a $2/3$-competitive algorithm for MCM on growing trees, which proves \Cref{thm:grwoingtreealg}. Then we construct a hard instance to show an upper bound of $2/3$, which proves \Cref{thm:growingtreehard}, and thus, the algorithm we propose is optimal. 

\subsection{A \boldmath{\texorpdfstring{$2/3$}{a}}-Competitive Algorithm}

\noindent We only present the fractional algorithm in this section, and by \Cref{lem:rounding1} we can obtain an equivalent randomized integral algorithm. 

In our algorithm, edges can be divided into three types: unmatched, stable ($1/3$-matched), and unstable ($2/3$-matched or fully matched). The matched fraction of a stable edge will never change, while the matched fraction of an unstable edge may be decreased to $1/3$ in the future. When an edge $e$ arrives, if it is incident to another unstable edge $e'$ then we make $e'$ stable. After that we greedily match the maximum possible fraction to $e$. See \cref{alg:FracMCM} for a formal description.

\vspace{5pt}
\begin{algorithm}[ht]
    \caption{Fractional Algorithm for MCM on Growing Trees}
    \label{alg:FracMCM}
    When an edge $e = (u,v)$ arrives:\\
    \If {$u$ has an incident edge $e' = (u,z)$ such that $x(e') \geq 2/3$} {
        Make $u$ stable: $x(e') \leftarrow 1/3$\;
    }
    $x(e) \leftarrow 1-x(u)$\;
\end{algorithm}

We use the Online Primal Dual Framework to analyze the algorithm's competitive ratio. The following is a dual assignment rule that naturally satisfies reverse weak duality, and we will prove that approximate dual feasibility is satisfied for each edge under this rule.

\begin{prop}[Parent-only dual assignment rule]
\label{prop:parentonly}
    If an edge $(u,v)$ is matched for some fraction $x(u,v)$, where $u$ is $v$'s parent, then we assign $x(u,v)$ to $\alpha(u)$, while $\alpha(v)$ gets $0$. In other words, for each non-leaf vertex $u$, we have $\alpha(u) = x(u) - x(p_u, u)$. (if $u$ is the root, assume $x(p_u,u) = 0$)
\end{prop}

\begin{lem}
\label{lem:nonleaf_1/3}
    For each non-leaf vertex $u$, $\alpha(u) \geq 1/3$.
\end{lem}
\begin{proof}
    Suppose the first child of $u$ is $v$, then we match at least $2/3$ to $(u,v)$ when it arrives. Then no matter what happens in the future, $x(u,v) \geq 1/3$ always holds. Therefore by \cref{prop:parentonly} we have $\alpha(u) \geq x(u,v) \geq 1/3$.
\end{proof}

\begin{lem}
\label{lem:nonleaf_2/3}
    For each non-leaf vertex $u$, if it is adjacent to at least one leaf, then $\alpha(u) \geq 2/3$.
\end{lem}
\begin{proof}
    If $u$ has only one child $v$, then before we greedily match $(u,v)$, there is at most one stable edge incident to $u$. Therefore we match at least $2/3$ to $(u,v)$, so $\alpha(u) = x(u,v) \geq 2/3$.

    If $u$ has at least two children, suppose the first two arriving children are $v_1,v_2$. Similar to the previous case, we greedily match at least $2/3$ to $(u,v_1)$, and $x(u,v_1) \geq 1/3$ always holds afterward. Before we greedily match $v_2$, $u$ has at most $2$ stable incident edges, so we match at least $1/3$ to $(u,v_2)$, and $x(u,v_2) \geq 1/3$ always holds afterward. Therefore we have $\alpha(u) \geq x(u,v_1) + x(u,v_2) \geq 2/3$.
\end{proof}

\begin{lem}
\label{mcm_frac_2/3}
    \cref{alg:FracMCM} is $2/3$-competitive.
\end{lem}
\begin{proof}
    Consider the two constraints of the Online Primal Dual Framework:
    \begin{itemize}
        \item Reverse Weak Duality: $P = D$ is satisfied by \cref{prop:parentonly}.
        
        \item Approximate dual feasibility: for each edge $(u,v)$, if $v$ is not a leaf, then by \cref{lem:nonleaf_1/3}, $\alpha(u) + \alpha(v) \geq 1/3 + 1/3 = 2/3$; otherwise if $v$ is a leaf, then by \cref{lem:nonleaf_2/3}, $\alpha(u) + \alpha(v) \geq 2/3 + 0 = 2/3$.
    \end{itemize}
    Therefore \cref{alg:FracMCM} is $2/3$-competitive.
\end{proof}

Finally, we can conclude the proof of \Cref{thm:grwoingtreealg}.
\begin{proof}[Proof of \Cref{thm:grwoingtreealg}]
    We apply the rounding from \cref{sec:framework} to \cref{alg:FracMCM}, then by \cref{mcm_frac_2/3} and \cref{lem:rounding1} the resulting randomized algorithm is $2/3$-competitive.
\end{proof}

\subsection{Hard Instance for MCM on Growing Trees}
\label{sec:hard}

\begin{proof} [Proof of \Cref{thm:growingtreehard}]

We construct a family of input $\{T_n\mid  n\in \mathbb{Z}^+\}$. The structure and arrival order of $T_n$ is as follows:

\begin{itemize}
    \item Vertex set: $\{u_i\}_{i=1}^{n+1} \cup \{v_i\}_{i=1}^{n}$.
    \item Edge set: $\{(u_i,v_i)\}_{i=1}^{n}\cup \{(u_i,u_{i+1}\}_{i=1}^{n}$.
    \item Arrival order: $(u_1,v_1), (u_1,u_2), (u_2,v_2), (u_2,v_3)\cdots (u_n,v_n), (u_n, u_{n+1})$.
\end{itemize}

Suppose after $(u_i,v_i)$ and $(u_i,u_{i+1})$ have both arrived, the fractional algorithm matches $x_i$ and $y_i$ to $(u_i,u_{i+1})$ and $(u_i,v_i)$ respectively. Here we can assume $x_i \geq y_i$, in other words, the adversary chooses to grow the tree on the child of $u_i$ that is matched a larger fraction. An example with $n = 3$ is illustrated in \cref{mcmhard}.

Besides, when the two children of $u_{i+1}$ arrive, it is not beneficial to dispose any matched fraction from $(u_{i},u_{i+1})$, because no subsequent edge will arrive at $u_{i}$. We assume that no \emph{reasonable} algorithm will do this, so in the end $(u_{i}, u_{i+1})$ still has matched fraction $x_{i}$. Only when the second child of $u_i$ arrives, a \emph{reasonable} algorithm will dispose some fraction matched to the first child. Next we prove that the competitive ratio of any reasonable fractional algorithm on $T_n$ is at most $\frac{2}{3} + \frac{2}{3n}$.

\begin{figure}[ht]
    \centering{
        \begin{tikzpicture}[sibling distance = 3cm]
            \node[node] {$u_1$}
                child {node[node] {$v_1$}
                    edge from parent node[above left] {$y_1$}}
                child {node[node] {$u_2$}
                  child {node[node] {$v_2$}
                    edge from parent node[above left] {$y_2$}}
                  child {node[node] {$u_3$}
                    child {node[node] {$v_3$}
                        edge from parent node[above left] {$y_3$}}
                    child {node[node] {$u_4$}
                        edge from parent node[above right] {$x_3$}}
                    edge from parent node[above right] {$x_2$}
                  }
                  edge from parent node[above right] {$x_1$}
                };
        \end{tikzpicture}
    } 
    \caption{Hard Instance for MCM}
    \label{mcmhard}
\end{figure}
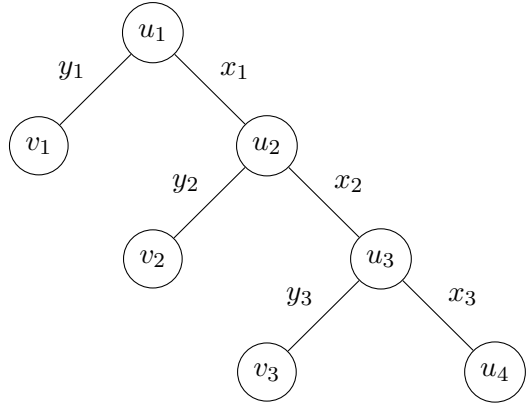

Any fractional algorithm must satisfy the following constraints: for each $u_i, i =2, 3, \cdots n$, the sum of matched fractions of the three incident edges does not exceed $1$, therefore
$$
x_{i-1}+x_i+y_i \leq 1 ~.
$$

Summing both sides of all constraints over $i=2,3,\cdots n$, we have
$$
x_1+ 2\sum_{i=2}^{n-1}x_i + x_n + \sum_{i=2}^{n}y_i \leq n-1 ~.
$$

The optimal solution to this instance is to match all edge $(u_1,v_1)\cdots (u_n, v_n)$, so $\opt = n$. The total matched fraction of the algorithm is $\alg = \sum_{i=1}^{n}(x_i+y_i)$. Since we assume $x_i\geq y_i$ for all $i\in[n]$, $\sum_{i=1}^{n}x_i \geq \frac{1}{2}\alg$, then
$$
\begin{aligned}
\frac{3}{2}\alg \leq &~\sum_{i=1}^{n}(2x_i+y_i)\\
\leq &~ n-1 + (x_1+y_1)+x_n\\
\leq &~ n-1 + 1 + 1\\
=&~ \opt + 1 ~.
\end{aligned}
$$
Therefore, we can upper bound the competitive ratio by $\frac{\alg}{\opt} \leq \frac{2}{3}+\frac{2}{3n}$. Thus when $n \to \infty$, we can conclude the proof of \Cref{thm:growingtreehard}.

\end{proof}

\section{Online MCM on Forests}
\label{sec:forest}

\noindent In this section, we present a $5/8$-competitive fractional algorithm for MCM with free disposal on Forests, which satisfies the conditions in \Cref{lem:rounding2}.

This algorithm is based on the $2/3$-competitive algorithm for Growing Trees in the previous section, which works fine if no edges that merge two trees arrive. But problems start to occur when an arriving edge merges two trees. For example, consider two leaves $u$ and $v$, following \Cref{alg:FracMCM} and the parent-only dual assignment rule, they can be both $2/3$-matched with $\alpha(u)=\alpha(v)=0$. However, if an edge $e = (u,v)$ arrives, from the dual perspective, we cannot afford to cover $e$ because the primal solution can increase by at most $1/3$. To overcome this challenge, we introduce a surplus dual assignment rule on each tree, which means we mainatin 
$$
P = \sum_{(u,v)\in E} x(u,v) \geq \sum_{u\in V} \alpha(u) + C = D + C ~. 
$$
As a result, we target to prove $\alg \geq c \cdot \opt + C$. It can be done by carefully handling the root vertex because of the special property that the root does not have a parent. In our $5/8$-competitive algorithm, we always maintain a $5/8$-approximate dual solution on each tree, with a $C=2/8$ amount of surplus as follows. 

\begin{inv}[Surplus Invariant]
\label{inv:2/8surplus}
    For each independent tree, we have a condition stronger than Reverse Weak Duality: $P \geq D+2/8$. 
\end{inv}

Note that the surplus does not affect the competitive ratio on each tree because it is only a constant. However, when we merge trees, we have two copies of surplus on each tree but only need to maintain one copy after merging. Therefore, we can use one copy of the surplus to ensure the approximate dual feasibility of the arriving edge.

\paragraph{Types of Edges.} Before introducing our algorithm, we classify edges into three types.
    A \emph{Growing Edge} is an edge that "grows" a leaf to a tree. For a growing edge $(u,v)$, unless otherwise specified, we assume that $u$ is already in the graph and $v$ is a new vertex, and call $u$ the parent of $v$. 
    A \emph{Non-Trivial Merging Edge} is an edge that merges two trees (with at least two edges). Conversely, if  one of the trees only consists of one edge, the merging edge is called a
    \emph{Trivial Merging Edge}. Note that we will view the trivial merging edges as growing edges but with a specific rule. 
    
\paragraph{Types of Non-leaf Vertices.} Corresponding to the edge types, we classify non-leaf vertices into two types. A non-leaf vertex is \emph{type A} if all its incident edges with non-zero matched fractions are growing edges. Otherwise, a non-leaf vertex is \emph{type B} if at least one of its incident (non-trivial) merging edges has a non-zero matched fraction. Our matching strategy and dual assignment rule always keep the type of a non-leaf vertex unchanged. 

\begin{inv}
\label{inv:fixtype}
    A non-leaf vertex is fixed to be type A or type B the first time it becomes a non-leaf vertex. 
\end{inv}

\subsection{Matching Rule for a Growing Edge to a Type A Vertex or a Leaf}
\label{sec:typeA}

\noindent In this section, we introduce how we match a growing edge to a type A vertex or a Leaf, which is almost aligned with \Cref{alg:FracMCM}. To recall, the difference is that we need to carefully match edges incident to the root to gain a $2/8$-surplus for dual assignments. In particular, here, the notion of "root" is slightly different from the Growing Trees setting. 

For convenience, the matched fractions of each edge is always a multiple of $1/8$, and we will always use $8$ as the denominator when describing matched fractions. We call a growing edge $e$ stable if $x(e) \leq 3/8$ and unstable if $x(e) \geq 4/8$. When a growing edge $(u,v)$ arrives, we make each unstable edge adjacent to $u$ (the parent) stable by lowering its matched fraction to $3/8$. Notice that when $u$ is the root, it may have two incident $4/8$-matched edges. We present details in \Cref{alg:FracMCM2}.

\vspace{5pt}
\begin{algorithm}[ht]
    \caption{Matching Rule for A Growing Edge to A Type A Vertex or A Leaf}
    \label{alg:FracMCM2}
    When an growing edge $e = (u,v)$ arrives, and the parent $u$ is currently type A or a Leaf\;
    \uIf {$e$ is not incident to any previous edge} {
        $x(e) \leftarrow 1$\;
    } \uElseIf {$e$ is incident to $e'$, and $e'$ is not incident to any previous edge} {
        $x(e') \leftarrow 4/8$\;
        $x(e) \leftarrow 4/8$\;
        regard the common vertex of $e$ and $e'$ as a "root".
    } \Else {
        \For {each growing edge $e' = (u,z)$ or $e' = (z,u)$ where $z$ is also type A, such that $x(e') \geq 4/8$} {
            $x(e') \leftarrow 3/8$\;
        }
        $x(e) \leftarrow 1-x(u)$\;
    }
    Update dual assignments using \Cref{prop:parent2}
\end{algorithm}

Similar to the analysis of \Cref{alg:FracMCM}, we use a parent-only dual assignment rule (\Cref{prop:parent2}) for each growing edge (also applies if $u$ is type B). To precisely formalize the surplus, we need a specific rule for a root. 

\begin{prop}[Parent-only Dual Assignment Rule]
\label{prop:parent2}
The dual assignment rule for growing edges are:
\begin{itemize}
    \item For each growing edge $(u,v)$, where $u$ is the parent, $\alpha(u)$ gets $x(u,v)$ from $(u,v)$ and $\alpha(v)$ gets nothing. 
    \item For the root vertex $r$, which must be type A, we remove $2/8$ from its dual as the surplus. That is, $\alpha(r) = \sum_{u:(r,u)\in E} x(r,u) - 2/8$.
\end{itemize}    
\end{prop}

The motivation of the matching strategy and the dual assignment rule is to maintain the following invariant for Type A vertices.

\begin{inv} [Type A Invariant]
\label{inv:typeA}
For each type A vertex $u$, $\alpha(u) \geq 3/8$, and if $u$ is adjacent to at least one leaf, then $\alpha(u) \geq 5/8$.
\end{inv}

Note that \Cref{inv:typeA} is already sufficient to prove the $5/8$-competitive ratio if assuming the Growing Trees setting.

\begin{prop}
\label{inv:stable}
     A growing edge $e$ (both to type A and to type B) is stable if $x(e) \leq 3/8$. The following invariants corresponding to the stable property should be kept.
     \begin{enumerate}
         \item If a growing edge $e$ is stable, we will never decrease $x(e)$ in the future. 
         \item If a growing edge $e$ is unstable, we may decrease $x(e)$ to make it stable. But we always have $x(e)\geq 3/8$. 
     \end{enumerate}
\end{prop}

\begin{lem}
\label{lem:safedisposal}
    A growing edge $(u,v)$ to a type A vertex $u$ can only be unstable when $v$ is a leaf. When $v$ becomes a non-leaf, decreasing $(u,v)$ to $3/8$ (make it stable) can inductively mainatin \Cref{inv:typeA} for $u$. 
\end{lem}
\begin{proof}
  If $u$ is a root, then $u$ has an unstable incident edge if $u$ has only two children. Recall that $u$ has two $4/8$-matched incident leaves at the beginning. There are two cases:
  \begin{itemize}
      \item If $v$ is the only $4/8$-matched leaf incident to $u$, then the other child $v'$ is already a non-leaf and $x(u,v') = 3/8$. Now $v$ also becomes a non-leaf, and the disposal step updates $x(u,v) \leftarrow 3/8$, therefore $\alpha(u) = x(u,v) + x(u,v') - 2/8 = 4/8 > 3/8$, which satisfies \Cref{inv:typeA} since $u$ is not incident to any leaf after $v$ becomes a non-leaf.
      \item If $u$ has two $4/8$-matched incident leaves $v,v'$, then after updating $x(u,v) \leftarrow 3/8$, we have $\alpha(u) = x(u,v) + x(u,v') - 2/8 = 5/8$, which satisfies \Cref{inv:typeA}.
  \end{itemize}

  If $u$ is not a root, there are also two cases, depending on whether $v$ is the only child of $u$:
  \begin{itemize}
      \item If $v$ is the only child of $u$, then after $v$ becomes non-leaf, we dispose $x(u,v)$ to $3/8$. By the parent-only dual assignment rule, we have $\alpha(u) \geq 3/8$, so \Cref{inv:typeA} is satisfied because $u$ is no longer adjacent to any leaf.
      \item If $u$ has another child $v'$, then $v'$ must arrive before $v$ (otherwise $(u,v)$ cannot be unstable). Further, $(u,v')$ must be unstable after $v'$ has just arrived. Otherwise, we couldn't have matched any fraction to $(u,v)$. Therefore after $v$ arrives, $x(u,v') = 3/8$. Thus, after $v$ becomes a non-leaf and $x(u,v) \leftarrow 3/8$, by the parent-only dual assignment rule, we have $\alpha(u) = 6/8$, which satisfies \Cref{inv:typeA}.  \qedhere
  \end{itemize}
\end{proof}

\begin{lem}
    After a growing edge to a type A vertex arrives, by \Cref{alg:FracMCM2} and the dual assignment rule, \Cref{inv:typeA} and \Cref{inv:2/8surplus} hold inductively. 
\end{lem}
\begin{proof}
    The base case is when a tree has only two edges, then for the root $u$ we have $\alpha(u) = x(u) - 2/8 = 6/8$.

    Now suppose the arriving edge is $(u,v)$. If $u$ is a root, then $x(u) = 1$ must hold after $(u,v)$ arrives, so by the special dual assignment rule, we have $\alpha(u) = x(u) - 2/8 = 6/8 \geq 5/8$. If $u$ is not a root, we have shown in \Cref{lem:safedisposal} that \Cref{inv:typeA} can be kept for $p_u$. For $u$, \Cref{alg:FracMCM2} always ensures $x(u) = 1$ and $x(p_u,u) \leq 3/8$ after $(u,v)$ arrives, so by the parent-only dual assignment rule, we have $\alpha(u) = x(u) - x(p_u,u) \geq 5/8$. Therefore \Cref{inv:typeA} holds inductively after $(u,v)$ arrives.

    For each growing edge $(u,v)$, whenever we update $x(u,v)$ in \Cref{alg:FracMCM2}, we do the same update to $\alpha(u)$, therefore  \Cref{inv:2/8surplus} holds inductively. 
\end{proof}
    
\paragraph{Remark for Trivial Merging edges.} The case that a tree to be merged has only two vertices is special: suppose a tree has only two vertices $u_1,v_1$. The arrival edge is $(v_1,v_2)$, and $v_2$ is in another tree. We can view $(v_1,v_2)$ arriving first as a growing edge, and then $(u_1,v_1)$ arriving. \Cref{alg:FracMCM2} shows that we have matched $1$ to $(u_1,v_1)$, so the matched fraction of $(u_1,v_1)$ can be achieved by disposal.

\subsection{Type B Vertices and Merging Edges}
\label{sec:typeB}

\noindent Before describing specific rules, we introduce some basic invariants we plan to maintain for type B vertices. We define a stable invariant aligned with \Cref{inv:stable}.
\begin{inv}
\label{inv:stableB}
Generally speaking, a vertex becomes stable when it becomes type B. We have the following invariants:
\begin{enumerate}
    \item The dual value of a type B vertex $u$, i.e., $\alpha(u)$, never decreases. 
    \item All growing edges adjacent to a type B vertex are stable. 
\end{enumerate}
\end{inv} 

On the other hand, we do not have the strong property that $\alpha(u) \geq 3/8$ as type A vertices. Instead, we introduce the following invariant. 
\begin{inv}
\label{inv:unsafe}
    For each type B vertex, $\alpha(u) \geq 2/8$. We call it safe if $\alpha(u) + 1 - x(u) \geq 5/8$. If $u$ is unsafe, we have $\alpha(u) + 1 - x(u) = 4/8$, and there must be an adjacent edge $(u,v)$ where we have an uncertain dual assignment of $1/8$ stored on $(u,v)$, and $v$ is also unsafe. 
\end{inv}

We remark that the rule for growing edges to type A vertices never affects the dual value of type B vertices. So, \Cref{inv:stableB} and \Cref{inv:unsafe} inductively hold by these rules.

\paragraph{Uncertain Dual Assignment.} Here, the uncertain dual assignment is a special assignment rule when we handle merging edges. For example, we may match a $3/8$-fraction to a merging edge $(u,v)$, but we do not fix the gain of $3/8$ to the dual of $u$ or $v$ immediately. There may be a $1/8$ amount of uncertain dual value temporarily stored on this edge. As a result, we may have $\alpha(u) = \alpha(v) = 2/8$ with an uncertain dual assignment of $1/8$ on $(u,v)$. This edge still has the approximate dual feasibility because $2/8+2/8+1/8 = 5/8$.

\paragraph{} Next we formally present the algorithms related to Type B vertices and merging edges. \cref{alg:FracMCMgotoB} describes the rules for growing an edge to a type B vertex. Here if the Type B vertex is unsafe then we remove the corresponding uncertain assignment and make it safe. \cref{alg:FracMCMmerge} describes the rules for merging edges. If the required invariants can all be satisfied using the surplus dual, then we do not match any fraction to the merging edge. This addresses the following 2 cases: 
\begin{itemize}
    \item both endpoints are non-leaf.
    \item one endpoint is a leaf and the other is non-leaf with dual at least $3/8$. 
\end{itemize}
Otherwise we must match some fraction to the merging edge, which covers the remaining 2 cases: 
\begin{itemize}
    \item both endpoints are leaves, this is where we create an uncertain dual assignment.
    \item one endpoint is a leaf and the other is non-leaf with dual $2/8$.
\end{itemize}

\subsubsection{Growing a Leaf to a Type B Vertex}

\noindent When a growing edge $(u,v)$ arrives where $u$ is a Type B vertex, the rules are listed as follows in \Cref{alg:FracMCMgotoB}. Intuitively, if $u$ is safe, by $\alpha(u) + 1 - x(u) \geq 5/8$, we can provide enough gain to $\alpha(u)$ to make it equal to $5/8$, by matching $u$ with a fraction of $5/8-\alpha(u)$. However, if $u$ is unsafe, we only have $\alpha(u) + 1 - x(u) \geq 4/8$, the unmatched portion of $u$ may not be enough. 
To fix this, we first dispose $1/8$ from $(u,z)$, the edge that stores an uncertain dual assignment, which must exist by \Cref{inv:unsafe}, and also remove the $1/8$ uncertain assignment. As a result, both $u$ and $z$ become safe, which means we inductively maintain the unsafe requirement in \Cref{inv:unsafe}. On the other hand, the unmatched portion of $u$ becomes enough. 

\begin{algorithm}[ht]
    \caption{Matching Rule for a Growing Edge to a Type B Vertex}
    \label{alg:FracMCMgotoB}
    When an growing edge $e = (u,v)$ arrives, and the parent $u$ is currently type B:\\
    \If{$u$ is unsafe} {
        Find the edge $(u,z)$ with uncertain assignment\;
        $x(u,z) \gets x(u,z) - 1/8$\;
        Remove the uncertain assignment on $(u,z)$\;
    }
    $x(u,v) \gets \max \{5/8 - \alpha(u),~0\}$\; 
    Update dual assignments using \Cref{prop:parent2}\;
\end{algorithm}

\begin{prop}
\label{prop:typeBstable}
    If $(u,v)$ is a growing edge to a type B vertex, it must be stable. 
\end{prop}
\begin{proof}
    Because $x(u,v) = \max \{5/8 - \alpha(u),~0\}$ and $\alpha(u) \geq 2/8$. Therefore $x(u,v) \leq 3/8$ and it is stable.  
\end{proof}

\begin{lem}
    After a growing edge to a type B vertex arrives, 
    \Cref{inv:2/8surplus}, 
    \Cref{inv:fixtype},  \Cref{inv:typeA}, \Cref{inv:stableB}, \Cref{inv:unsafe} holds. 
\end{lem}
\begin{proof}
    \Cref{inv:fixtype} and \Cref{inv:typeA} are trivially not affected. In \Cref{alg:FracMCMgotoB}, if we dispose $1/8$ from edge $(u,z)$, we also remove the uncertain assignment stored on $(u,z)$, while no dual of any vertex is decreased, so \Cref{inv:2/8surplus} and \Cref{inv:stableB} both hold. Besides, both $u$ and $z$ become safe after the disposal, so \Cref{inv:unsafe} is also maintained.
\end{proof}

The following lemma shows a sufficient gain of dual value for the type B vertex $u$.
\begin{lem} 
\label{lem:growB_5/8}
    If a type B vertex $u$ is the parent of a growing edge, then $\alpha(u) \geq 5/8$.
\end{lem}
\begin{proof}
    Notice that in \Cref{alg:FracMCMgotoB}, the first time a growing edge arrives to a Type B vertex $u$, we match $\max\{5/8-\alpha(u),0\}$ to that edge and thus increase $\alpha(u)$ to $5/8$. And by \Cref{inv:stableB}, $\alpha(u)$ will not decrease afterwards.
\end{proof}

\subsubsection{Merging Edges: No Matching Increase}
\label{sec:merging_rules}

\noindent This section discusses the cases in which we do not match any fraction to the arriving merging edge $(u,v)$.
To recall, we do not change the type of a non-leaf vertex when handling merging edges. It means that if $u$ or $v$ is type A, we cannot match any fraction to $(u,v)$. 

This requirement can be satisfied by the following rules: when a merging edge merges two non-leaf vertices $u,v$, or merges a leaf $u$ with a non-leaf vertex $v$ with $\alpha(v) \geq 3/8$, we do not match any fraction to $(u,v)$. To maintain the invariants of dual variables, we update the dual of $u$ and $v$ as follows: 
\begin{enumerate}
    \item For a non-leaf vertex, if it is type B and has dual $2/8$, we increase it to $3/8$.
    \item For a leaf $u$, we increase $\alpha(u)$ to $2/8$.
\end{enumerate}

\begin{lem}
    For a merging edge $(u,v)$, if $x(u,v) = 0$ at its arrival, we can inductively maintain 
    \Cref{inv:2/8surplus}, 
    \Cref{inv:fixtype}  \Cref{inv:typeA}, \Cref{inv:stableB}, \Cref{inv:unsafe}.
\end{lem}
\begin{proof}
    In both rules, the total increase of dual is at most $2/8$, so \Cref{inv:2/8surplus} still holds. The other invariants hold trivially.
\end{proof}

\begin{lem}
\label{lem:mergenomatch}
    For a merging edge $(u,v)$, if $x(u,v) = 0$ at its arrival, we have $\alpha(u)+\alpha(v) \geq 5/8$.
\end{lem}
\begin{proof}
    After updating the dual variables, for the first rule, we have $\alpha(u) \geq 3/8$ and $\alpha(v)\geq 3/8$; for the second rule, we have $\alpha(u) = 2/8$ and $\alpha(v) \geq 3/8$. Therefore $\alpha(u) + \alpha(v) \geq 5/8$ in both rules after the arrival of edge $(u,v)$.
\end{proof}

\subsubsection{Merging Edges: Matching Increase and Creating Type B Vertices}
The remaining cases are (1) merging two leaves. (2) merging a leaf with a type B vertex $v$, where $\alpha(v) =2/8$. In both cases, we first perform a disposal step:
\begin{itemize}
    \item Consider all leaf vertices among $u$ and $v$. If their corresponding edge $e$ is unstable, decrease $x(e)$ to $3/8$. 
\end{itemize}

After such a disposal step, a leaf is matched $\leq 3/8$, and we can safely assume that each leaf is matched exactly $3/8$. Therefore, when designing merging rules that involve a leaf, we only need to consider $3/8$-matched leaves.

When merging two $3/8$-matched leaves, we create the uncertain dual assignment mentioned above, as well as two unsafe Type B vertices. When merging a leaf with a type B vertex, using one $2/8$-surplus is not enough to satisfy approximate dual feasibility, so we have to match some fraction to the merging edge. If the type B vertex here is unsafe, we also need to dispose $1/8$ from the corresponding edge and remove the uncertain dual assignment. All details are shown in \Cref{alg:FracMCMmerge}. \Cref{forest1}, \Cref{forest2}, \Cref{forest3} provide graphical illustrations to these merging rules.

\begin{algorithm}[ht]
    \caption{Merging Edges: Matching and Dual Assignment}
    \label{alg:FracMCMmerge}
    When a merging edge $e = (u,v)$ arrives:\\
    \uIf {$\alpha(u) + \alpha(v) \geq 3/8$} {
        $x(u,v) \gets 0$\;
        if one of $u,v$ is a leaf, increase its dual to $2/8$\;
        for type B vertices with $2/8$ dual among $u,v$, increase the dual to $3/8$\;
    } 
    \Else {
        \uIf {$u$ is leaf and $x(u) > 3/8$} {
            $\alpha(p_u) \gets \alpha(p_u) - (x(u) - 3/8)$\;
            $x(u) \gets 3/8$\;
        }
        \uIf {$v$ is leaf and $x(v) > 3/8$} {
            $\alpha(p_v) \gets \alpha(p_v) - (x(v) - 3/8)$\;
            $x(v) \gets 3/8$\;
        }
        \uIf {$u$ and $v$ are both leaves} {
            $x(u,v) \gets 3/8$\;
            $\alpha(u) \gets 2/8$\; $\alpha(v) \gets 2/8$\;
            Assign an uncertain dual value of $1/8$ on $(u,v)$\;
        } \ElseIf{one of them is a leaf} {
            Assume $u$ is the leaf and $v$ is the type B non-leaf\;
            \uIf {$v$ is unsafe} {
                Find the edge $(v,z)$ with uncertain assignment\;
                $x(v,z)\gets x(v,z) -1/8$\;
                Remove the uncertain dual assignment on $(v,z)$\; 
            }
            $x(u,v) \gets 1/8$\;
            $\alpha(u) \gets 2/8$\;
            $\alpha(v) \gets 3/8$\;
        }
    }
\end{algorithm}

\begin{lem}
    For a merging edge $(u,v)$, if $x(u,v) > 0$ at its arrival, we can inductively maintain 
    \Cref{inv:2/8surplus}, 
    \Cref{inv:fixtype}  \Cref{inv:typeA}, \Cref{inv:stableB}, \Cref{inv:unsafe}.
\end{lem}
\begin{proof}
    The only step involving a Type A vertex is when merging a leaf $u$ with $x(u) > 3/8$, we need to decrease $x(u)$ to $3/8$, but by \Cref{lem:safedisposal}, \Cref{inv:typeA} still holds for $p_u$.

    It can be verified that for each case, we have $\Delta D \leq \Delta P + 2/8$, so we only need to spend one $2/8$-surplus to fill in the gap, and the merged tree still has a $2/8$-surplus, so \Cref{inv:2/8surplus} still hold.

    In all cases of such a merging edge, \Cref{alg:FracMCMmerge} either creates a Type B vertex and assigns a $2/8$-dual, or increases the dual of a Type B vertex and only matches a $\leq 3/8$-fraction to the merging edge, so \Cref{inv:fixtype} and \Cref{inv:stableB} are maintained.

    Finally, if the merging edge involves an unsafe Type B vertex $v$, suppose the corresponding edge with uncertain dual assignment is $(v,z)$, then both $v$ and $z$ become safe, so \Cref{inv:unsafe} is preserved.
\end{proof}

\begin{lem}
\label{lem:merging_adf}
    For a merging edge $(u,v)$, if $x(u,v) > 0$ at its arrival, then approximate dual feasibility is satisfied on $(u,v)$ at this time.
\end{lem}
\begin{proof}
    If $u,v$ are both leaves, we assign $\alpha(u) = \alpha(v) = 2/8$, but there is an uncertain $1/8$-dual stored on $(u,v)$, so when considering approximate dual feasibility, we may temporarily assign this $1/8$-dual to $\alpha(u)$ or $\alpha(v)$ to ensure $\alpha(u)+\alpha(v) = 5/8$. If $u$ is a leaf and $v$ is non-leaf, we assign $\alpha(u) = 2/8$ and $\alpha(v) = 3/8$, so $\alpha(u)+\alpha(v) = 5/8$. 
\end{proof}

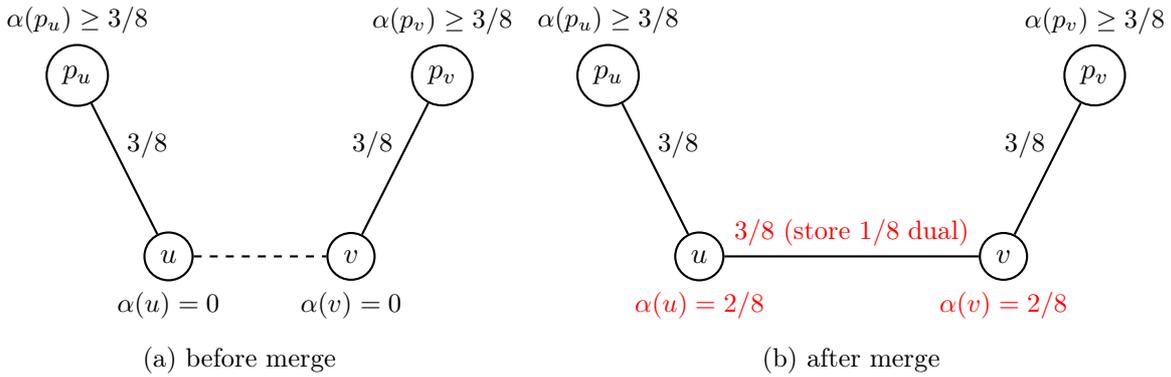
\begin{figure}[ht]
\begin{subfigure}{0.4\textwidth}
    \centering{
    \begin{tikzpicture}[scale=0.8, auto, node distance=1cm, thick]
        \node[node] (pu) [label={[black]90:\small$\alpha({p_u})\geq3/8$}] at (-3,3) {$p_u$};
        \node[node] (u) [label={[black]270:\small$\alpha({u})=0$}] at (-1.5,0) {$u$};
        \node[node] (pv) [label={[black]90:\small$\alpha({p_v})\geq3/8$}] at (3,3) {$p_v$};
        \node[node] (v) [label={[black]270:\small$\alpha({v})=0$}] at (1.5,0) {$v$};
        \draw[edge] (pu) to node[above,black,xshift=3mm] {\small$3/8$} (u);
        \draw[edge] (pv) to node[above,black,xshift=-3mm] {\small$3/8$} (v);
        \draw[dashed] (u) to node {} (v);
    \end{tikzpicture}
    \caption{before merge}
    }
\end{subfigure}
\begin{subfigure}{0.6\textwidth}
    \centering{
    \begin{tikzpicture}[scale=0.8, auto, node distance=1cm, thick]
        \node[node] (pu) [label={[black]90:\small$\alpha(p_u)\geq3/8$}] at (-4,3) {$p_u$};
        \node[node] (u) [label={[red]270:\small$\alpha(u)=2/8$}] at (-2.5,0) {$u$};
        \node[node] (p_v) [label={[black]90:\small$\alpha(p_v)\geq3/8$}] at (4,3) {$p_v$};
        \node[node] (v) [label={[red]270:\small$\alpha(v)=2/8$}] at (2.5,0) {$v$};
        \draw[edge] (pu) to node[above,black,xshift=3mm] {\small$3/8$} (u);
        \draw[edge] (p_v) to node[above,black,xshift=-3mm] {\small$3/8$} (v);
        \draw[edge] (u) to node[above,red] {\small{$3/8$ (store $1/8$ dual)}} (v);
    \end{tikzpicture}
    \caption{after merge}
    }
\end{subfigure}
\caption{Merging Two Leaves}
\label{forest1}
\end{figure}

\begin{figure}[ht]
\begin{subfigure}{0.5\textwidth}
    \centering{
    \begin{tikzpicture}[scale=0.8, auto, node distance=1cm, thick]
        \node[node] (pu) [label={[black]90:\small$\alpha({p_u})\geq3/8$}] at (-3,3) {$p_u$};
        \node[node] (u) [label={[black]270:\small$\alpha({u})=0$}] at (-1.5,0) {$u$};
        \node[node] (p_v) at (0,3) {$p_v$};
        \node[node] (v) [label={[black]270:\small$\alpha({v})=2/8$}] at (1.5,0) {$v$};
        \node[node] (z) at (4.5,0) {$z$};
        \draw[edge] (pu) to node[above,black,xshift=3mm] {\small$3/8$} (u);
        \draw[edge] (p_v) to (v);
        \draw[dashed] (u) to node {} (v);
        \draw[edge] (v) to (z);
    \end{tikzpicture}
    \caption{before merge}
    }
\end{subfigure}
\begin{subfigure}{0.5\textwidth}
    \centering{
    \begin{tikzpicture}[scale=0.8, auto, node distance=1cm, thick]
        \node[node] (pu) [label={[black]90:\small$\alpha({p_u})\geq3/8$}] at (-3,3) {$p_u$};
        \node[node] (u) [label={[red]270:\small$\alpha({u})=2/8$}] at (-1.5,0) {$u$};
        \node[node] (p_v) at (0,3) {$p_v$};
        \node[node] (v) [label={[red]270:\small$\alpha({v})=3/8$}] at (1.5,0) {$v$};
        \node[node] (z) at (4.5,0) {$z$};
        \draw[edge] (pu) to node[above,black,xshift=3mm] {\small$3/8$} (u);
        \draw[edge] (p_v) to (v);
        \draw[dashed] (u) to node[above,red] {\small$1/8$} (v);
        \draw[edge] (v) to (z);
    \end{tikzpicture}
    \caption{after merge}
    }
\end{subfigure}
\caption{Merging a Leaf with a Safe Type B Vertex with $2/8$-Dual}
\label{forest2}
\end{figure}
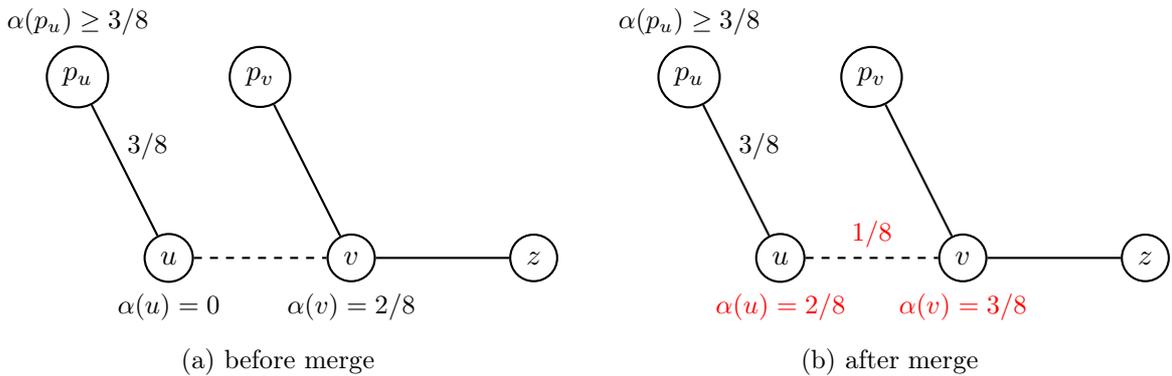

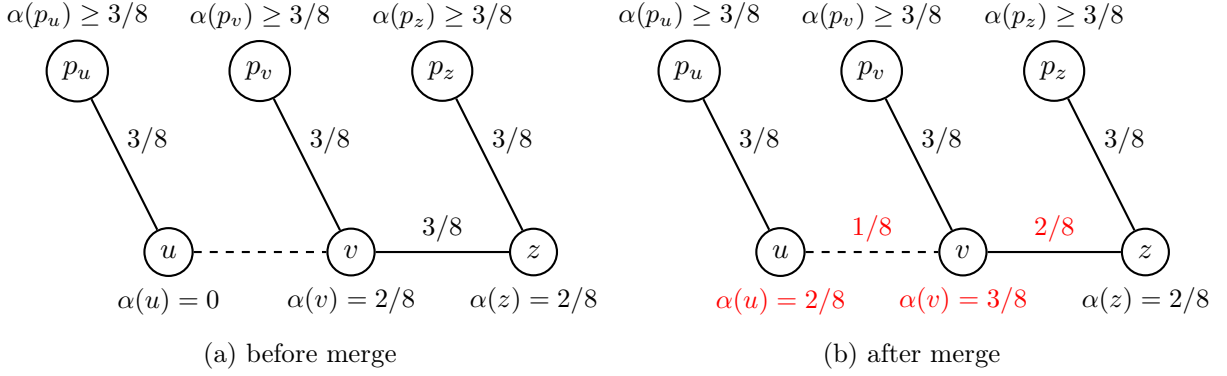
\begin{figure}[ht]
\begin{subfigure}{0.5\textwidth}
    \centering{
    \begin{tikzpicture}[scale=0.8, auto, node distance=1cm, thick]
        \node[node] (pu) [label={[black]90:\small$\alpha({p_u})\geq 3/8$}] at (-3,3) {$p_u$};
        \node[node] (u) [label={[black]270:\small$\alpha({u})=0$}] at (-1.5,0) {$u$};
        \node[node] (p_v) [label={[black]90:\small$\alpha({p_v})\geq 3/8$}] at (0,3) {$p_v$};
        \node[node] (v) [label={[black]270:\small$\alpha({v})=2/8$}] at (1.5,0) {$v$};
        \node[node] (pz) [label={[black]90:\small$\alpha({p_z})\geq 3/8$}] at (3,3) {$p_z$};
        \node[node] (z) [label={[black]270:\small$\alpha({z})=2/8$}] at (4.5,0) {$z$};
        \draw[edge] (pu) to node[above,black,xshift=3mm] {\small$3/8$} (u);
        \draw[dashed] (u) to node {} (v);
        \draw[edge] (p_v) to node[above,black,xshift=3mm] {\small$3/8$} (v);
        \draw[edge] (v) to node[above,black] {\small$3/8$} (z);
        \draw[edge] (pz) to node[above,black,xshift=3mm] {\small$3/8$} (z);
    \end{tikzpicture}
    \caption{before merge}
    }
\end{subfigure}
\begin{subfigure}{0.5\textwidth}
    \centering{
    \begin{tikzpicture}[scale=0.8, auto, node distance=1cm, thick]
        \node[node] (pu) [label={[black]90:\small$\alpha({p_u})\geq 3/8$}] at (-3,3) {$p_u$};
        \node[node] (u) [label={[red]270:\small$\alpha({u})=2/8$}] at (-1.5,0) {$u$};
        \node[node] (p_v) [label={[black]90:\small$\alpha({p_v})\geq 3/8$}] at (0,3) {$p_v$};
        \node[node] (v) [label={[red]270:\small$\alpha({v})=3/8$}] at (1.5,0) {$v$};
        \node[node] (pz) [label={[black]90:\small$\alpha({p_z})\geq 3/8$}] at (3,3) {$p_z$};
        \node[node] (z) [label={[black]270:\small$\alpha({z})=2/8$}] at (4.5,0) {$z$};
        \draw[edge] (pu) to node[above,black,xshift=3mm] {\small$3/8$} (u);
        \draw[dashed] (u) to node[above,red] {\small$1/8$} (v);
        \draw[edge] (p_v) to node[above,black,xshift=3mm] {\small$3/8$} (v);
        \draw[edge] (v) to node[above,red] {\small$2/8$} (z);
        \draw[edge] (pz) to node[above,black,xshift=3mm] {\small$3/8$} (z);
    \end{tikzpicture}
    \caption{after merge}
    }
\end{subfigure}
\caption{Merging a Leaf with an Unsafe Type B Vertex}
\label{forest3}
\end{figure}

\subsection{Put Everything Together}

\begin{lem}
\label{lem_frac5/8}
    There is a $5/8$-competitive fractional algorithm for online MCM on forests.
\end{lem}

\begin{proof}
    Consider the two constraints of the Online Primal Dual Framework:
    \begin{itemize}
        \item Reverse Weak Duality: by \Cref{inv:2/8surplus} we have $P \geq D + 2/8$ for each tree of the input graph.
        \item Approximate Dual Feasibility: considering all combinations of endpoints of edges,
        \begin{enumerate}
            \item Both endpoints are Type A: satisfied by \Cref{inv:typeA}.
            \item One endpoint is a leaf: the other endpoint has dual $\geq 5/8$ by \Cref{inv:typeA} (Type A) or \Cref{lem:growB_5/8} (Type B).
            \item One endpoint is Type A, and the other is Type B: by \Cref{inv:typeA} and \Cref{inv:unsafe}, the sum of dual is at least $3/8 + 2/8 = 5/8$.
            \item Both endpoints are Type B: by \Cref{lem:mergenomatch} and \Cref{lem:merging_adf}. The dual feasibility is maintained after we make a matching decision on its arrival time. Because the dual variable of type B vertices never decreases by \Cref{inv:stableB}, we only need to consider the case in which we lose an uncertain dual assignment. However, it only happens when we make an unsafe vertex safe. In this case, we can check that one of them will gain a dual variable of $3/8$. We also have $\alpha(u) + \alpha(v) \geq 3/8 + 2/8 \geq 5/8$.
        \end{enumerate}
    \end{itemize}
    
    Therefore, combining \Cref{alg:FracMCM2}, \Cref{alg:FracMCMgotoB} and \Cref{alg:FracMCMmerge} yields a $5/8$-competitive fractional algorithm for online MCM on forests. 
\end{proof}

Finally, we apply the transformation from \Cref{sec:framework} and conclude the proof of \Cref{thm:forestealg}.

\begin{proof} [Proof of \Cref{thm:forestealg}]
    To apply the rounding scheme, it remains to prove that the conditions in \Cref{lem:rounding2} are satisfied. 
    From \Cref{alg:FracMCM2} it is direct to see $x(u,v)=1$ when $(u,v)$ is an isolated edge.
    Second, regarding the inequality $\frac{\gamma}{(1-\mu_1)(1-\mu_2)} \leq 1$ in \Cref{lem:rounding2}, we only need to check two triples of $(\gamma, \mu_1,\mu_2)$: 
    \begin{itemize}
        \item $(3/8,3/8,3/8)$ (when merging two leaves).
        \item $(1/8,3/8,5/8)$ (when merging a leaf and a Type B vertex with $2/8$ dual).
    \end{itemize} 
    Both triples satisfy $\frac{\gamma}{(1-\mu_1)(1-\mu_2)} \leq 1$. Other cases are either trivial ($\gamma = 0$) or strictly easier than the two triples above (such as $(3/8, 2/8,3/8)$).

    Therefore by \Cref{lem:rounding2}, the transformation from \Cref{sec:framework} yields a $5/8$-competitive randomized algorithm for online MCM with disposal on Forests.
\end{proof}

\section{Online MWM on Growing Trees}
\label{sec:weighted}

\subsection{A Simple Ordinal Algorithm}

Following the general framework, we first maintain a fractional matching, which, for the sake of simplicity, only assigns $1/2$-fractions to edges. If an edge has a $1/2$-fraction, we call it half-matched for short. Decisions of disposal only depend on the comparative relationship between the weights of the two edges. Each vertex can have at most two half-matched incident edges, and when the third edge arrives, the algorithm will dispose of the edge with the lowest weight. The details are presented in \cref{alg:FracMWM}.

\vspace{5pt}
\begin{algorithm}[h]
    \caption{Fractional Algorithm for MWM on Growing Trees}
    \label{alg:FracMWM}
    When edge $e = (u,v)$ arrives:\\
    \eIf{$m(u) \leq 1/2$}{
        $x(e) \leftarrow 1/2$\;
    }{
        Suppose $e_1 = (u,x)$ and $e_2 = (u,y)$ are half-matched,\\ 
        Wlog. suppose $w(e_1) \leq w(e_2)$.\\
        \eIf{$w(e) > w(e_1)$} {
            $x(e_1) \leftarrow 0$\;
            $x(e) \leftarrow 1/2$\;
        }{
            $x(e) \leftarrow 0$\;
        }
    }
\end{algorithm}

Unlike the analysis for MCM, this time, we do not maintain the dual variables or the matching. Instead, we show that it is always possible to construct a set of dual variables that satisfy the two constraints of the Online Primal Dual Framework.

Let $U$ denote the set of non-leaf vertices. For each $u \in U$, define 
$$
h_u = \argmax_{v\in child(u)}\{w(u,v)\} ~,
$$
i.e., the child of $u$ with the largest edge weight. To break ties, choose the child that arrives the earliest. We refer to edges of the form $(u,h_u)$ as heavy edges.

\begin{lem}
\label{mwm_frac_lem1}
    Let $M_F$ be the subset of edges $e\in E$ such that $x(e) = 1/2$, then at any time,
    $$
    \sum_{e\in M_F}w(e) \geq \sum_{u\in U} w(u, h_u) ~.
    $$
\end{lem}
\begin{proof}
For each $u\in U$, we will charge the heavy edge $(u, h_u)$ to a unique $\mu(u) \in M_F$ such that $w(\mu(u)) \geq w(u, h_u)$, thus proving the inequality.

When $(u, h_u)$ first arrives, there cannot be two half-matched edges incident to $u$ that are heavier than $(u, h_u)$, so \cref{alg:FracMWM} always matches $1/2$-fraction to it. If $(u, h_u)$ is not disposed, we simply set $\mu(u) = (u, h_u)$.

However, if $(u, h_u)$ has been disposed of, we need to charge $u$ to an edge that is "responsible" for the disposal of $(u,h_u)$. That edge may also have been disposed of, and we need to charge further. For each edge $(u, v)$, we stipulate how it should be charged when it is disposed of:

\begin{itemize}
    \item Case 1: $(u, v)$ is disposed when an edge $(v, x)$ arrives, then there must be another half-matched edge $(v, y)$, such that $w(v, y) \geq w(u, v)$ and $w(v, x) \geq w(u, v)$. 
    
    In this case, if $x \neq h_v$, then charge $(u, v)$ to $(v, x)$. Otherwise $y \neq h_v$, then charge $(u, v)$ to $(v, y)$.
    
    \item Case 2: $(u, v)$ is disposed when an edge $(u, x)$ arrives, and there is another half-matched edge $(u, y)$, such that $w(u, y) \geq w(u, v)$ and $w(u, x) \geq w(u, v)$. 
    
    In this case, if $x \neq h_u$, then charge $(u, v)$ to $(u, x)$. Otherwise $y \neq h_u$, then charge $(u, v)$ to $(u, y)$.

    \item Case 3: $(u, v)$ is disposed when an edge $(u, x)$ arrives, and the edge between $u$ and its parent is half-matched. We will show that this case is impossible.
\end{itemize}

For each $u\in U$ such that $(u, h_u)$ is disposed of in the end, since $(u, h_u)$ is the heaviest among all edges to children of $u$, it falls into Case 1. Suppose the chain of charging is $(u_1,v_1) \to (u_2,v_2)\to (u_3,v_3)\cdots \to (u_k,v_k)$, where $(u_1, v_1) = (u, h_u)$. 

If when $(u_i, v_i)$ is disposed of, the edge between $u_i$ and its parent has already been disposed of, then the disposal of $(u_i, v_i)$ must fall into Case 1 or Case 2, and in either case, the edge between $u_{i+1}$ and its parents must be disposed of after the disposal of $(u_i, v_i)$. Since $(u, h_u)$ is disposed at the beginning, by induction, for each $i = 2, 3\cdots k-1$, when $(u_i, v_i)$ is disposed of, the edge between $u_i$ and its parent must have been disposed of, so the disposal of $(u_i, v_i)$ must fall into Case 1 or Case 2, thus the charging rule above is well-defined. 

Define the depth of an edge as the depth of its upper vertex on the growing tree. Since the depth of the edges on a chain of charging is non-decreasing, it must end at some $(u_k, v_k) \in M_F$, and we set $\mu(U) = (u_k, v_k)$. 

Finally, we prove that no two chains of charging have a common edge. We prove this by contradiction. Let $(u,v)$ be the edge with the smallest depth that belongs to at least two chains (if there are multiple such edges, choose one arbitrarily). Since both Case1 and Case2 avoid charging an edge to a heavy edge, $(u,v)$ must not be a heavy edge. Suppose two of the chains that contain $(u,v)$ are $C_1,C_2$, then $(u,p_u)$ does not belong to one of $C_1,C_2$ (due to the minimality of the depth of $(u,v)$), wlog. Suppose it does not belong to $C_1$. Since each pair of neighboring edges on a chain of charging must share a common vertex, we conclude that $C_1$ must start in the subtree of $u$. Therefore $C_1$ cannot contain $(u,v)$, because:
\begin{itemize}
    \item If $C_1$ starts at $(u,h_u) \neq (u,v)$, then the remaining part of $C_1$ is in the subtree of $h_u$ and thus cannot contain $(u,v)$.
    \item If $C_1$ does not start at $(u,h_u)$, then it obviously cannot contain $(u,v)$ either.
\end{itemize}
This contradicts the assumption that $(u,v)$ belongs to $C_1$. Therefore, such $(u,v)$ does not exist, which implies that no edge belongs to two chains of charging. Hence for $u_1,u_2 \in U, u_1\neq u_2$, we can ensure that $\mu(u_1) \neq \mu(u_2)$.
\end{proof}

\begin{lem}
\label{mwm_frac_1/2}
    \cref{alg:FracMWM} is $1/2$-competitive.
\end{lem}
\begin{proof}
    We assign $\alpha(u) = \frac{1}{2}w(u, h_u)$ for each $u\in U$, and for each leaf $v\not\in U$, simply set $\alpha(v) = 0$. Consider the two constraints of the Online Primal Dual Framework:
    \begin{enumerate}
        \item Reverse weak duality: by \cref{mwm_frac_lem1}, $P = \frac{1}{2}\sum_{e\in M_F}w(e) \geq \frac{1}{2}\sum_{u\in U}w(u,h_u) = D$.
        \item Approximate dual feasibility: for each edge $(u,v)$ where $u$ is the parent, we have $\alpha(u)+\alpha(v) \geq \alpha(u) = \frac{1}{2}\cdot w(u, h_u) \geq \frac{1}{2}\cdot w(u,v)$.
    \end{enumerate}
    Therefore \cref{alg:FracMWM} is $1/2$-competitive. The ratio of $1/2$ is tight because we only match a $1/2$-fraction to the first arriving edge.
\end{proof}

Finally, we conclude the proof of \Cref{thm:weightedalg}.

\begin{proof} [Proof of \Cref{thm:weightedalg}]
    We apply the transformation from \cref{sec:framework} to \cref{alg:FracMWM}, then by \cref{mwm_frac_1/2} and \cref{lem:rounding1} the resulting randomized algorithm is $1/2$-competitive.
\end{proof}

\subsubsection{Implementation Using One Bit of Randomness}

As a final remark, our randomized algorithm for MWM can also be implemented to use only one bit of randomness. 

Let $E_F$ denote the set of edges that are either half-matched or were half-matched and later disposed. In each connected component of the induced subgraph of $E_F$, the first arriving edge in the component is matched with probability $1/2$, then all subsequent edges in the component are in fact matched deterministically (it can be verified that the corresponding probabilities are all $1$). Since components are independent of each other, it does not matter if they use a shared random bit instead of generating a new random bit every time.

Therefore, we can describe our randomized algorithm in a way similar to that of \cite{Tirodkar1/3}: maintain two matchings $M_0$ and $M_1$, such that $M_0$ never matches the first edge in a component while $M_1$ always matches the first edge in a component. Use one bit of global randomness to decide which matching to output.

\subsection{Upper Bound for Ordinal Algorithms}

\noindent \Cref{alg:FracMWM} is very simple, and we conjecture that better algorithms exist when making full use of the values of edge weights. Nevertheless, we are able to show that it is optimal among ordinal algorithms.

\begin{proof} [Proof of \Cref{thm:weightedhard}]
    The high-level idea of our proof is to construct two instances $\mathcal{I}_1$ and $\mathcal{I}_2$ that are indistinguishable for an ordinal algorithm and show that any fractional ordinal algorithm cannot be better than $1/2$-competitive on both instances.

    The two instances share the same structure as $T_n$ defined in the proof for \Cref{thm:grwoingtreealg} in \Cref{sec:hard}. For edges $(u_i,v_i), i=1\cdots n$, we assign weight $1+(i-1)\epsilon$ in $\mathcal{I}_1$ and $C^{i-1}$ in $\mathcal{I}_2$. For edges $(u_i,u_{i+1}),i=1\cdots n$, we assign weight $1+i\epsilon$ in $\mathcal{I}_1$ and $C^i$ in $\mathcal{I}_2$. An example with $n=3$ is illustrated in \cref{mwmhard}.

    \begin{figure}[ht]
    \begin{subfigure}{0.5\textwidth}
        \centering{
        \begin{tikzpicture}[sibling distance = 3cm]
            \node[node] {$u_1$}
                child {node[node] {$v_1$}
                    edge from parent node[above left] {$1$}}
                child {node[node] {$u_2$}
                  child {node[node] {$v_2$}
                    edge from parent node[above left] {$1+\epsilon$}}
                  child {node[node] {$u_3$}
                    child {node[node] {$v_3$}
                        edge from parent node[above left] {$1+2\epsilon$}}
                    child {node[node] {$u_4$}
                        edge from parent node[above right] {$1+3\epsilon$}}
                    edge from parent node[above right] {$1+2\epsilon$}
                  }
                  edge from parent node[above right] {$1+\epsilon$}
                };
        \end{tikzpicture}
        \caption{$\mathcal{I}_1$}
        }
    \end{subfigure}
    \begin{subfigure}{0.5\textwidth}
        \centering{
        \begin{tikzpicture}[sibling distance = 3cm]
            \node[node] {$u_1$}
                child {node[node] {$v_1$}
                    edge from parent node[above left] {$1$}}
                child {node[node] {$u_2$}
                  child {node[node] {$v_2$}
                    edge from parent node[above left] {$C$}}
                  child {node[node] {$u_3$}
                    child {node[node] {$v_3$}
                        edge from parent node[above left] {$C^2$}}
                    child {node[node] {$u_4$}
                        edge from parent node[above right] {$C^3$}}
                    edge from parent node[above right] {$C^2$}
                  }
                  edge from parent node[above right] {$C$}
                };
        \end{tikzpicture}
        \caption{$\mathcal{I}_2$}
        }
    \end{subfigure}
    \caption{Hard Instance for MWM}
    \label{mwmhard}
    \end{figure}
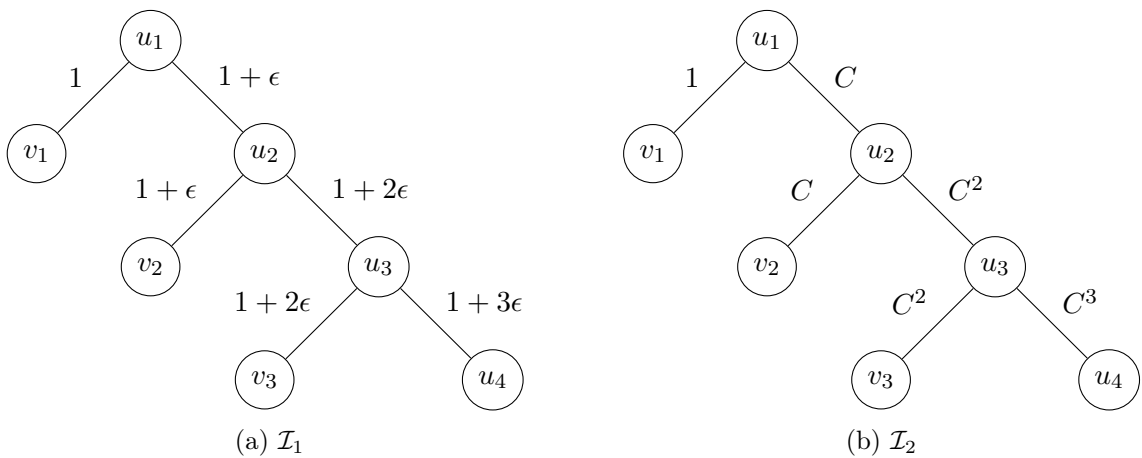

    Since for every pair of edges, the comparative relationship between their edge weights are the same for $\mathcal{I}_1$ and $\mathcal{I}_2$, any fractional ordinal algorithm cannot distinguish between the two instances and should make the same decisions. Next we show that for all $\delta \in (0,1/2]$, there exists $n, C, \epsilon$ such that any $1/2+\delta$ competitive ordinal algorithm for $\mathcal{I}_2$ is less than $1/2$-competitive for $\mathcal{I}_1$.

    For some sufficiently large $C\in \mathbb{R}^+$, when each edge $(u_i,u_{i+1})$ arrives, its weight in $\mathcal{I}_2$ dominates the sum of all previously arrived edges, so any $1/2+\delta$ competitive algorithm $\alg$ for $\mathcal{I}_2$ must match at least a $1/2+\delta$ fraction to $(u_i,u_{i+1})$. This implies that after all edges arrived, for each $i = 1\cdots n-1$, the total matched fraction of $(u_i, u_{i+1})$ and $(u_{i+1}, v_{i+1})$ is at most $1/2-\delta$. Thus, we can bound the performance of $\alg$ for $\mathcal{I}_1$:
    $$
    \begin{aligned}
    \alg(\mathcal{I}_1) \leq &~ 1 + (1 + n\epsilon) + (1/2-\delta)\cdot\sum_{i=1}^{n-1}(1+i\epsilon)\\
    =&~ 2+n\epsilon + (1/2-\delta)\left(n-1+\frac{n^2-n}{2}\cdot \epsilon\right) ~.
    \end{aligned}
    $$
    The optimal offline algorithm matches $\{(u_i,v_i)\}_{i=1}^{n-1} \cup \{(u_n,u_{n+1})\}$, thus
    $$
    \begin{aligned}
    \opt(\mathcal{I}_1) = &~ (1+n\epsilon) + \sum_{i=1}^{n-1}(1+(i-1)\epsilon) \\
    =&~ n+\frac{n^2-n+2}{2}\cdot \epsilon ~.
    \end{aligned}
    $$
    Therefore, the competitive ratio of $\alg$ for $\mathcal{I}_1$ is upper bounded by
    $$
    \dfrac{\alg(\mathcal{I}_1)}{\opt(\mathcal{I}_1)} \leq (1/2-\delta) + \dfrac{3/2+\delta + (n-1)\epsilon}{n+\frac{n^2-n+2}{2}\cdot \epsilon} = (1/2-\delta)+O(n^{-1}) ~,
    $$
    so for sufficiently large $n$ and small constant $\epsilon$, $\alg$ is less than $1/2$-competitive, which concludes the whole proof.
\end{proof}

\newpage
\bibliographystyle{plainnat}
\bibliography{main}

\begin{thebibliography}{34}
\providecommand{\natexlab}[1]{#1}
\providecommand{\url}[1]{\texttt{#1}}
\expandafter\ifx\csname urlstyle\endcsname\relax
  \providecommand{\doi}[1]{doi: #1}\else
  \providecommand{\doi}{doi: \begingroup \urlstyle{rm}\Url}\fi

\bibitem[Aggarwal et~al.(2011)Aggarwal, Goel, Karande, and Mehta]{vertexweight}
Gagan Aggarwal, Gagan Goel, Chinmay Karande, and Aranyak Mehta.
\newblock Online vertex-weighted bipartite matching and single-bid budgeted allocations.
\newblock In \emph{{SODA}}, pages 1253--1264. {SIAM}, 2011.

\bibitem[Ashlagi et~al.(2023)Ashlagi, Burq, Dutta, Jaillet, Saberi, and Sholley]{windowed}
Itai Ashlagi, Maximilien Burq, Chinmoy Dutta, Patrick Jaillet, Amin Saberi, and Chris Sholley.
\newblock Edge-weighted online windowed matching.
\newblock \emph{Math. Oper. Res.}, 48\penalty0 (2):\penalty0 999--1016, 2023.

\bibitem[Babaioff et~al.(2018)Babaioff, Immorlica, Kempe, and Kleinberg]{matroidsecretary}
Moshe Babaioff, Nicole Immorlica, David Kempe, and Robert Kleinberg.
\newblock Matroid secretary problems.
\newblock \emph{J. {ACM}}, 65\penalty0 (6):\penalty0 35:1--35:26, 2018.

\bibitem[Buchbinder and Naor(2009)]{pdframework}
Niv Buchbinder and Joseph Naor.
\newblock The design of competitive online algorithms via a primal-dual approach.
\newblock \emph{Found. Trends Theor. Comput. Sci.}, 3\penalty0 (2-3):\penalty0 93--263, 2009.

\bibitem[Buchbinder et~al.(2014)Buchbinder, Jain, and Singh]{JKsecretary}
Niv Buchbinder, Kamal Jain, and Mohit Singh.
\newblock Secretary problems via linear programming.
\newblock \emph{Math. Oper. Res.}, 39\penalty0 (1):\penalty0 190--206, 2014.

\bibitem[Buchbinder et~al.(2019)Buchbinder, Segev, and Tkach]{Buchbinder5/9}
Niv Buchbinder, Danny Segev, and Yevgeny Tkach.
\newblock Online algorithms for maximum cardinality matching with edge arrivals.
\newblock \emph{Algorithmica}, 81\penalty0 (5):\penalty0 1781--1799, 2019.

\bibitem[Buchbinder et~al.(2023)Buchbinder, Naor, and Wajc]{lossless_rounding}
Niv Buchbinder, Joseph~(Seffi) Naor, and David Wajc.
\newblock Lossless online rounding for online bipartite matching (despite its impossibility).
\newblock In \emph{{SODA}}, pages 2030--2068. {SIAM}, 2023.

\bibitem[Bullinger and Romen(2023)]{Bullinger1/6}
Martin Bullinger and Ren{\'{e}} Romen.
\newblock Online coalition formation under random arrival or coalition dissolution.
\newblock In \emph{{ESA}}, volume 274 of \emph{LIPIcs}, pages 27:1--27:18. Schloss Dagstuhl - Leibniz-Zentrum f{\"{u}}r Informatik, 2023.

\bibitem[Chiplunkar et~al.(2015)Chiplunkar, Tirodkar, and Vishwanathan]{Chiplunkar15/28}
Ashish Chiplunkar, Sumedh Tirodkar, and Sundar Vishwanathan.
\newblock On randomized algorithms for matching in the online preemptive model.
\newblock In \emph{{ESA}}, volume 9294 of \emph{Lecture Notes in Computer Science}, pages 325--336. Springer, 2015.

\bibitem[Correa et~al.(2022)Correa, Cristi, Epstein, and Soto]{googol}
Jos{\'{e}} Correa, Andr{\'{e}}s Cristi, Boris Epstein, and Jos{\'{e}}~A. Soto.
\newblock The two-sided game of googol.
\newblock \emph{J. Mach. Learn. Res.}, 23:\penalty0 113:1--113:37, 2022.

\bibitem[Devanur et~al.(2013)Devanur, Jain, and Kleinberg]{Devanur13}
Nikhil~R. Devanur, Kamal Jain, and Robert~D. Kleinberg.
\newblock Randomized primal-dual analysis of {RANKING} for online bipartite matching.
\newblock In \emph{{SODA}}, pages 101--107. {SIAM}, 2013.

\bibitem[Epstein et~al.(2013)Epstein, Levin, Segev, and Weimann]{Epstein5.356}
Leah Epstein, Asaf Levin, Danny Segev, and Oren Weimann.
\newblock Improved bounds for online preemptive matching.
\newblock In \emph{{STACS}}, volume~20 of \emph{LIPIcs}, pages 389--399. Schloss Dagstuhl - Leibniz-Zentrum f{\"{u}}r Informatik, 2013.

\bibitem[Fahrbach et~al.(2020)Fahrbach, Huang, Tao, and Zadimoghaddam]{weighted0.5086}
Matthew Fahrbach, Zhiyi Huang, Runzhou Tao, and Morteza Zadimoghaddam.
\newblock Edge-weighted online bipartite matching.
\newblock In \emph{{FOCS}}, pages 412--423. {IEEE}, 2020.

\bibitem[Feigenbaum et~al.(2005)Feigenbaum, Kannan, McGregor, Suri, and Zhang]{Feigenbaum05}
Joan Feigenbaum, Sampath Kannan, Andrew McGregor, Siddharth Suri, and Jian Zhang.
\newblock On graph problems in a semi-streaming model.
\newblock \emph{Theor. Comput. Sci.}, 348\penalty0 (2-3):\penalty0 207--216, 2005.

\bibitem[Gamlath et~al.(2019)Gamlath, Kapralov, Maggiori, Svensson, and Wajc]{Gamlath19}
Buddhima Gamlath, Michael Kapralov, Andreas Maggiori, Ola Svensson, and David Wajc.
\newblock Online matching with general arrivals.
\newblock In \emph{{FOCS}}, pages 26--37. {IEEE} Computer Society, 2019.

\bibitem[Gao et~al.(2021)Gao, He, Huang, Nie, Yuan, and Zhong]{weighted0.519}
Ruiquan Gao, Zhongtian He, Zhiyi Huang, Zipei Nie, Bijun Yuan, and Yan Zhong.
\newblock Improved online correlated selection.
\newblock In \emph{{FOCS}}, pages 1265--1276. {IEEE}, 2021.

\bibitem[Goel and Mehta(2008)]{Adwords2}
Gagan Goel and Aranyak Mehta.
\newblock Online budgeted matching in random input models with applications to adwords.
\newblock In \emph{{SODA}}, pages 982--991. {SIAM}, 2008.

\bibitem[Gravin et~al.(2023)Gravin, Sun, and Tang]{ordinal}
Nick Gravin, Enze Sun, and Zhihao~Gavin Tang.
\newblock Online ordinal problems: Optimality of comparison-based algorithms and their cardinal complexity.
\newblock In \emph{{FOCS}}, pages 1863--1876. {IEEE}, 2023.

\bibitem[Huang et~al.(2018)Huang, Kang, Tang, Wu, Zhang, and Zhu]{fully18}
Zhiyi Huang, Ning Kang, Zhihao~Gavin Tang, Xiaowei Wu, Yuhao Zhang, and Xue Zhu.
\newblock How to match when all vertices arrive online.
\newblock In \emph{{STOC}}, pages 17--29. {ACM}, 2018.

\bibitem[Huang et~al.(2019{\natexlab{a}})Huang, Peng, Tang, Tao, Wu, and Zhang]{fully19}
Zhiyi Huang, Binghui Peng, Zhihao~Gavin Tang, Runzhou Tao, Xiaowei Wu, and Yuhao Zhang.
\newblock Tight competitive ratios of classic matching algorithms in the fully online model.
\newblock In \emph{{SODA}}, pages 2875--2886. {SIAM}, 2019{\natexlab{a}}.

\bibitem[Huang et~al.(2019{\natexlab{b}})Huang, Tang, Wu, and Zhang]{vertexweight2}
Zhiyi Huang, Zhihao~Gavin Tang, Xiaowei Wu, and Yuhao Zhang.
\newblock Online vertex-weighted bipartite matching: Beating 1-1/\emph{e} with random arrivals.
\newblock \emph{{ACM} Trans. Algorithms}, 15\penalty0 (3):\penalty0 38:1--38:15, 2019{\natexlab{b}}.

\bibitem[Huang et~al.(2020{\natexlab{a}})Huang, Tang, Wu, and Zhang]{fully20}
Zhiyi Huang, Zhihao~Gavin Tang, Xiaowei Wu, and Yuhao Zhang.
\newblock Fully online matching {II:} beating ranking and water-filling.
\newblock In \emph{{FOCS}}, pages 1380--1391. {IEEE}, 2020{\natexlab{a}}.

\bibitem[Huang et~al.(2020{\natexlab{b}})Huang, Zhang, and Zhang]{Adwords3}
Zhiyi Huang, Qiankun Zhang, and Yuhao Zhang.
\newblock Adwords in a panorama.
\newblock In \emph{{FOCS}}, pages 1416--1426. {IEEE}, 2020{\natexlab{b}}.

\bibitem[Karp et~al.(1990)Karp, Vazirani, and Vazirani]{KVV}
Richard~M. Karp, Umesh~V. Vazirani, and Vijay~V. Vazirani.
\newblock An optimal algorithm for on-line bipartite matching.
\newblock In \emph{{STOC}}, pages 352--358. {ACM}, 1990.

\bibitem[McGregor(2005)]{McGregor05}
Andrew McGregor.
\newblock Finding graph matchings in data streams.
\newblock In \emph{{APPROX-RANDOM}}, volume 3624 of \emph{Lecture Notes in Computer Science}, pages 170--181. Springer, 2005.

\bibitem[Mehta et~al.(2007)Mehta, Saberi, Vazirani, and Vazirani]{Adwords}
Aranyak Mehta, Amin Saberi, Umesh~V. Vazirani, and Vijay~V. Vazirani.
\newblock Adwords and generalized online matching.
\newblock \emph{J. {ACM}}, 54\penalty0 (5):\penalty0 22, 2007.

\bibitem[Nuti and Vondr{\'{a}}k(2023)]{secretarysample}
Pranav Nuti and Jan Vondr{\'{a}}k.
\newblock Secretary problems: The power of a single sample.
\newblock In \emph{{SODA}}, pages 2015--2029. {SIAM}, 2023.

\bibitem[Qiu et~al.(2023)Qiu, Feng, Zhou, and Wu]{stochastic0.650}
Guoliang Qiu, Yilong Feng, Shengwei Zhou, and Xiaowei Wu.
\newblock Improved competitive ratio for edge-weighted online stochastic matching.
\newblock In \emph{{WINE}}, volume 14413 of \emph{Lecture Notes in Computer Science}, pages 527--544. Springer, 2023.

\bibitem[Soto et~al.(2021)Soto, Turkieltaub, and Verdugo]{matroidsecretary2}
Jos{\'{e}}~A. Soto, Abner Turkieltaub, and Victor Verdugo.
\newblock Strong algorithms for the ordinal matroid secretary problem.
\newblock \emph{Math. Oper. Res.}, 46\penalty0 (2):\penalty0 642--673, 2021.

\bibitem[Tang and Zhang(2022)]{fully24}
Zhihao~Gavin Tang and Yuhao Zhang.
\newblock Improved bounds for fractional online matching problems.
\newblock \emph{CoRR}, abs/2202.02948, 2022.

\bibitem[Tirodkar and Vishwanathan(2017)]{Tirodkar1/3}
Sumedh Tirodkar and Sundar Vishwanathan.
\newblock Maximum matching on trees in the online preemptive and the incremental dynamic graph models.
\newblock In \emph{{COCOON}}, volume 10392 of \emph{Lecture Notes in Computer Science}, pages 504--515. Springer, 2017.

\bibitem[Varadaraja(2011)]{Varadaraja11}
Ashwinkumar~Badanidiyuru Varadaraja.
\newblock Buyback problem - approximate matroid intersection with cancellation costs.
\newblock In \emph{{ICALP} {(1)}}, volume 6755 of \emph{Lecture Notes in Computer Science}, pages 379--390. Springer, 2011.

\bibitem[Wang and Wong(2015)]{WangWong}
Yajun Wang and Sam~Chiu{-}wai Wong.
\newblock Two-sided online bipartite matching and vertex cover: Beating the greedy algorithm.
\newblock In \emph{{ICALP} {(1)}}, volume 9134 of \emph{Lecture Notes in Computer Science}, pages 1070--1081. Springer, 2015.

\bibitem[Yan(2024)]{stochastic0.645}
Shuyi Yan.
\newblock Edge-weighted online stochastic matching: Beating 1-1/e.
\newblock In \emph{{SODA}}, pages 4631--4640. {SIAM}, 2024.

\end{thebibliography}

\end{document}